\documentclass[11pt]{amsart}

\usepackage{amsfonts}
\usepackage{graphicx}
\usepackage{amsmath}
\usepackage{amssymb}

 \advance\hoffset by -0.0 truecm

\newtheorem{theorem}{Theorem}

\newtheorem{lemma}{Lemma}

\theoremstyle{definition}
\newtheorem{definition}{Definition}
\theoremstyle{remark}
\newtheorem{remark}{Remark}
\numberwithin{equation}{section}

\newcommand{\re}{\hbox{Re}}

\DeclareMathOperator{\spn}{span}

\DeclareMathOperator{\tr}{tr}

\DeclareMathOperator{\pu}{\hbox{Pu}}

\begin{document}
\title[Geodesics on sub-Lorentzian manifolds]
{Geodesics on $\mathbb{H}$-type quaternion groups with sub-Lorentzian metric and their physical interpretation}

\author{Anna Korolko,\ Irina Markina}

\address{Department of Mathematics,
University of Bergen, Johannes Brunsgate 12, Bergen 5008, Norway}

\email{anna.korolko@uib.no}

\address{Department of Mathematics,
University of Bergen, Johannes Brunsgate 12, Bergen 5008, Norway}

\email{irina.markina@uib.no}

\thanks{The authors are partially supported by the grant of the Norwegian Research Council \# 177355/V30, by the grant of the European Science Foundation Networking Programme HCAA, and by the NordForsk Research Network Programme \# 080151}




\begin{abstract}

We study the existence and cardinality of normal geodesics of different causal types on $\mathbb {H}(eisenberg)$-type quaternion group equipped with the sub-Lorentzian metric. We present explicit formulas for geodesics and describe reachable sets by geodesics of different causal character. We compare results with the sub-Riemannian quaternion group and with the sub-Lorentzian Heisenberg group, showing that there are similarities and distinctions. We show that the geodesics on $\mathbb{H}$-type quaternion groups with the sub-Lorentzian metric satisfy the equations describing the motion of a relativistic particle in a constant homogeneous electromagnetic field. 
\end{abstract}

\maketitle

\section{Introduction}\label{sec:1}

The term {\it sub-Riemannian manifold} means the triple $(M,\mathcal H,d)$, where $M$ is an $n$-dimensional manifold, $\mathcal H$ is a smoothly varying $k$-dimensional distribution inside the tangent bundle $T(M)$ of the manifold $M$ with $k<n$, and $d$ is a Riemannian metric defined on $\mathcal H$, i.~e., a positively definite quadratic form. Recently the study of geometric structures, where the Riemannian metric $d$ on $\mathcal H$ is substituted by a semi-Riemannian metric $g$, that is a non-degenerate indefinite metric, started e.~g., in~\cite{Gr1,Gr2,Gr3,Gr4,Gr5,CMV,KM,KM2}. There is no special attribution so far for such kind of manifolds $(M,\mathcal H,g)$, thus we propose to call them {\it sub-semi-Riemannian} manifolds or shortly {\it ssr-manifolds}. In the particular case, when the metric $g$
has index 1, an ssr-manifold receives the name {\it sub-Lorentzian} manifold by the analogy to Lorentzian manifold.

In the present article we study an example of $\mathcal H$-type group furnished with the sub-Lorentzian metric. This is an interesting example not only as an almost unique known example of sub-Lorentzian manifold but also because it has a precise physical meaning. In the article we reveal the connection between sub-Lorentzian geometry and physics of relativistic electrodynamics basing on the example of  $H$-type quaternion group equipped with the Lorentzian metric. We also compare characterising features of sub-Riemannian and sub-Lorentzian geometries. The notion of $\mathcal H$-type groups was introduced in~\cite{Kap}. It is known that Riemannian manifolds have applications in classical mechanics. Sub-Riemannian manifolds of step 2 (Heisenberg manifolds) play important role in quantum mechanics. Sub-Riemannian geodesics even localy behave very differently from the ones in Riemannian geometry, where the energy minimising motion is described by a unique geodesic. A sub-atomic particle behaves in a way similar to an electron which moves only along a given set of directions. There can be infinitely many geodesics with different length joining two points. In its turn, the sub-Lorentzian structure underlise the motion in an electro-magnetic field. Just like space and time emerge in special relativity, the electric and magnetic fields can not be considered separately. The sub-Lorentzian structure absorbes both phenomena, the presence of the electro-magnetic field and the the space-time geometry. That is why the $\mathbb{H}$-type quaternion group with a sub-Lorentzian metric is an interesting example to work with.

In the introduction in order to explain the main idea we would like to mention the Heisenberg group as the simplest noncommutative example of $\mathbb H$-type (Heisenberg type) groups and its numerous connections with physics. The Heisenberg group is the manifold $\mathbb{R}^3$ with the noncommutative group law $$(x,y,z)\circ(x',y',z')=\big(x+x',y+y',z+z'+\frac{1}{2}(xy'-yx')\big).$$ Left-invariant vector fields $X=\partial_x-\frac{1}{2}y\partial_z$, $Y=\partial_y+\frac{1}{2}x\partial_z$ are obtained from the group law and span a 2-dimensional distribution $\mathcal H$ which is called horizontal. The horizontal distribution can be also defined as the kernel of the contact one-form $\omega=dz-\frac{1}{2}(xdy-ydx)$ in $\mathbb{R}^3$. The differential of $\omega$ is the 2-form $\Omega=d\omega=-dx\wedge dy$ that satisfies the Maxwell's equation $d\Omega=0$ for the magnetic field $\Omega$ in $\mathbb R^3$. Let us define a Riemannian metric $ds_{R}^2=dx^2+dy^2$ on $\mathcal H$. Then the sub-Riemannian manifold $\mathbb{H}^1_R=(\mathbb R^3,\mathcal H, ds_R^2)$ is also called the Heisenberg group. It turned out that the geodesic equation for geodesics $\gamma(t)$ satisfying the non-holonomic constraints $\dot\gamma(t)\in\mathcal H(\gamma(t))$ coincides up to a constant with the Lorentz equation of motion of the charged particle in the magnetic field $\Omega$. If we change Riemannian metric on $\mathcal H$ on the Lorentzian one $d{s}_L^2=-dx^2+dy^2$ we come to the notion of sub-Lorentzian Heisenberg group $\mathbb{H}^1_L=(\mathbb R^3,\mathcal H, ds_L^2)$. In this case the geodesic equation for non-holonomic geodesics coincides with analogue of the Lorentz equation for the motion of a charged particle in the electromagnetic field defined by $\Omega$ and the Lorentzian metric tensor. The geodesics, metrics properties, and other related questions on $\mathbb{H}^1_L$ were studied in~\cite{Gr1,Gr2,Gr3,KM}. The lack of the dimension of the horizontal distribution on the Heisenberg group does not allow to reveal the peculiarity of the applications in the case of the magnetic and electromagnetic fields. Therefore, we chose the analogue of the Heisenberg group admitting a 4-dimensional distribution, that we called quaternionic $\mathbb H$-type group. This example allows also to show similarities and differences between the sub-Riemannian and sub-Lorentzian geometries. 

The article is organized in the following way. In Section~\ref{sec:2} we introduce the $\mathbb H$-type quaternion groups endowed with different metrics: Riemannian and Lorentzian. We also present the differential equations for the geodesics in both cases. Section~\ref{sec:3} is the collection of the definitions related to the motion of charged particles in electro-magnetic fields. We give the explanation of the geometrical results from the physical point of view. Section~\ref{sec:4} is devoted to the solution of geodesic equations, where we fined the explicit formulas for the horizontal and vertical parts of geodesics. Section~\ref{sec:5} is dedicated to study of reachable sets by geodesics of different causal types and estimation of the cardinality of geodesics, connecting two different points. Section~\ref{sec:6} shows a brief overview of reachable sets for $\mathbb{H}^1_L$ for the sake of comparison with obtained results for the quaternion $\mathbb{H}$-type group.

\section{$\mathbb{H}$-type quaternion groups $\mathbb Q_R$ and $\mathbb Q_L$}\label{sec:2}

We remind that quaternions form a noncommutative division algebra that extends the system of complex numbers. It is convenient to define any quaternion $q$ in the algebraic form $q=a+i_1b+i_2c+i_3d$, where $i^2_1=i^2_2=i^2_3=i_1i_2i_3=-1$. The scalar $a\in\mathbb{R}$ is called the real part $\re \,q=a$ and a vector $(b,c,d)\in\mathbb{R}^3$ recieved the name pure imaginary quaternion and is denoted by $\pu q$. Thus $q=(\re \,q,\pu q)$. With this notations it is easy to introduce the noncommutative multiplication $*$ between two quaternions $q_1$ and $q_2$ using the usual inner product $\cdot$ and vector product $\times$ in $\mathbb R^3$. Namely, \begin{equation}\label{eq:qprod}q_1*q_2=(\re q_1\re q_1-\pu q_1\cdot\pu q_2,\ \re q_1\pu q_2+\re q_2\pu q_1+\pu q_1\times \pu q_2).\end{equation} Notice that this structure suggests itself the analogy with a Lorentzian geometry where $\mathbb{R}^4$ consists of a time part $t\in\mathbb{R}$ and space part $(x_1,x_2,x_3)\in \mathbb{R}^3$. The conjugate quaternion $\overline q$ to $q$ is $\overline{q}=a-i_1b-i_2c-i_3d$.
It is known that quaternion $q=a+i_1b+i_2c+i_3d$ can also be represented in the $4\times 4$-matrix form
\begin{gather*}
   q=\left(\begin{matrix}
            a& b &-d &-c\\
            -b& a &-c &d\\
            d& c &a &b\\
            c &-d &-b &a\\
      \end{matrix}\right)=aI+bI_1+cI_2+dI_3,
\end{gather*} where 
\begin{equation}\label{eq:Ibasis}
    \begin{split}
   &I=\left(\begin{matrix}
            1& 0 &0 &0\\
            0& 1 &0 &0\\
            0& 0 &1 &0\\
            0 &0 &0 &1\\
      \end{matrix}\right),\quad
      I_1=\left(\begin{matrix}
            \;0& 1 &\;0 &0\\
            -1& 0 &\;0 &0\\
            \;0& 0 &\;0 &1\\
            \;0 &0 &-1 &0\\
      \end{matrix}\right),\\
     &I_2=\left(\begin{matrix}
            0& 0 &\;0 &-1\\
            0& 0 &-1 &\;0\\
            0& 1 &\;0 &\;0\\
            1 &0 &\;0 &\;0\\
      \end{matrix}\right),\quad
     I_3=\left(\begin{matrix}
            0& \;0 &-1 &0\\
            0& \;0 &\;0 &1\\
            1& \;0 &\;0 &0\\
            0 &-1 &\;0 &0\\
      \end{matrix}\right)\end{split}
\end{equation} are the basis of quaternion numbers in the representation given by real $(4\times 4)$-matrices.

We introduce an $\mathbb H$-type group whose noncommutative multiplication law makes use of quaternion multiplication rule. Let us take the background manifold $M$ as $\mathbb{R}^7$ and define the noncommutative law by
\begin{equation}\label{eq:hprod}
(x,z)\circ (x',z')=\big(x+x',z+z'+\frac{1}{2}\pu (\bar x*x')\big)
\end{equation} for $(x,z)$ and $(x',z')$ from $\mathbb R^4\times\mathbb R^3$. Here $\pu (\bar x*x')$ is the imaginary part of the product $\bar x*x'$ defined in~\eqref{eq:qprod} of the conjugate quaternion $\bar x$ to $x$ by another quaternion $x'$. The introduced multiplication law~\eqref{eq:hprod} makes $M=(\mathbb{R}^7,\circ)$ into a noncommutative Lie group with the unity $(0,0)$ and the inverse element $(-x,-z)$ to $(x,z)$. The group law defines the left translation $L_{(x,z)}(x',z')=(x,z)\circ (x',z')$.
Let $\frac{\partial}{\partial x_0},\ldots,\frac{\partial}{\partial x_3},\frac{\partial}{\partial z_1},\ldots,\frac{\partial}{\partial z_3}$ be a standard basis of the tangent space $T_pM$ to $M$ at $p\in M$.
The basic left-invariant vector fields can be obtained by the action of the tangent map $dL_{(x,z)}$ of $L_{(x,z)}$ to the standard basis as $dL_{(x,z)}\big(\frac{\partial}{\partial x_i}\big)=X_i(x,z)$, $dL_{(x,z)}\big(\frac{\partial}{\partial z_k}\big)=Z_k(x,z)$. Then the vector fields
\begin{gather}
   X_0=\frac{\partial}{\partial x_0}+\frac{\displaystyle1}{\displaystyle2} \left(+x_1
   \frac{\partial}{\partial z_{1}}-x_3\frac{\partial}{\partial z_{2}}
   - x_2\frac{\partial}{\partial z_{3}}\right),
   \notag\\
   X_1=\frac{\partial}{\partial x_1}+\frac{\displaystyle1}{\displaystyle2} \left(-x_0
   \frac{\partial}{\partial z_{1}}-x_2\frac{\partial}{\partial z_{2}}
   + x_3\frac{\partial}{\partial z_{3}}\right),
   \notag\\
   X_2=\frac{\partial}{\partial x_2}+\frac{\displaystyle1}{\displaystyle2} \left(+x_3
   \frac{\partial}{\partial z_{1}}+x_1\frac{\partial}{\partial z_{2}}
   + x_0\frac{\partial}{\partial z_{3}}\right),
   \notag\\
   X_3=\frac{\partial}{\partial x_3}+\frac{\displaystyle1}{\displaystyle2} \left(-x_2
   \frac{\partial}{\partial z_{1}}+x_0\frac{\partial}{\partial z_{2}}
   - x_1\frac{\partial}{\partial z_{3}}\right),
   \notag
\end{gather} span a $4$-dimensional distribution $Q$, which we call {\it horizontal}.  The left-invariant vector fields $Z_{\beta}=\frac{\partial}{\partial z_{\beta}}$, $\beta=1,2,3$ form a basis of the complement $V$ to~$Q$ in the $TM$.  At each point $(x,z)\in \mathbb{R}^4\times\mathbb{R}^3$ the distribution $Q(x,z)$ is a copy of $\mathbb{R}^4$. 
The commutation relations are as follows
\begin{gather}
    [X_0,X_1]=-Z_1,\quad [X_0,X_2]=Z_3,\quad [X_0,X_3]=Z_2,\notag\\
    [X_1,X_2]=Z_2,\quad [X_1,X_3]=-Z_3,\quad [X_2,X_3]=-Z_1.\notag
\end{gather} Therefore, $\{X_0,\ldots,X_3\}$ and their commutators span the entire tangent space $TM$. This property of the distribution $Q$ is called bracket-generating of step 2. The Lie algebra with the basis $\{X_0,\ldots,X_3,Z_1,Z_2,Z_3\}$ is nilpotent of step 2. 

The horizontal distribution $Q$ can be defined by making use of one-forms. Namely, the one-forms
\begin{equation}\label{74}
\begin{split}
\omega_1 = & dz^{1}
-\frac{1}{2}(+x^1dx^0-x^0dx^1+x^3dx^2-x^2dx^3)=dz^{1}-\frac{1}{2}dx^TI_1
x,
\\ \omega_2 = & dz^{2}-\frac{1}{2}(-x^3dx^0-x^2dx^1+x^1dx^2+x^0dx^3)
=dz^{2}-\frac{1}{2}dx^TI_2 x,
\\ \omega_3 = & dz^{3}-\frac{1}{2}(-x^2dx^0+x^3dx^1+x^0dx^2-x^1dx^3)
=dz^{3}-\frac{1}{2}dx^TI_3 x.
\end{split}
\end{equation}
annihilate the distribution $Q$. Here $x=(x^0,\ldots,x^3)$, $dx=(dx^0,\ldots,dx^3)$, and $dx^T$ is the transposed vector to $dx$. Thus, $Q$ is the common kernel of forms $\omega_k$, $k=1,2,3$. Let us consider the external differential of the linear combination $\omega=\sum_{k=1}^{3}\vartheta_k\omega_k$. We get the two-form that is defined in 4-dimensional space
\begin{eqnarray} \Omega & = & \vartheta _1(dx^0\wedge dx^1+dx^2\wedge dx^3)+\vartheta_2(-dx^0\wedge dx^3-dx^1\wedge dx^2) \nonumber \\ & + & \vartheta_3(-dx^0\wedge dx^2+dx^1\wedge dx^3)
= \frac{1}{2}\sum_{ij}\Omega_{ij}dx^i\wedge dx^j,\nonumber \end{eqnarray}
where \begin{equation}\label{eq:omega}
   \Omega(\vartheta)= \Omega_{ij}(\vartheta)=\left(\begin{matrix}
            0& \vartheta_1 &-\vartheta_3 &-\vartheta_2\\
            -\vartheta_1& 0 &-\vartheta_2 &\vartheta_3\\
            \vartheta_3& \vartheta_2 &0 &\vartheta_1\\
            \vartheta_2 &-\vartheta_3 &-\vartheta_1 &0\\
      \end{matrix}\right).
\end{equation} 

Any vector $v\in Q_p$, $p\in M$, is called horizontal, a vector field $X$ tangent to $Q_p$ at each point $p$ is also called horizontal. An absolutely continuous curve $\gamma:[0,1]\to M$ that has its velocity vector $\dot \gamma(t)$ tangent to $Q_{\gamma(t)}$ almost everywhere is called horizontal curve.
  
\subsection{Sub-Riemannian manifold $\mathbb Q_R$}

\medskip

Let us define the Riemannian metric $(\cdot,\cdot)_R$ on the distribution $Q$ in such a way that $(X_i,X_j)=\delta_{ij}$, where $\delta_{ij}$ is the Kronecker symbol. With this we get the sub-Riemannian manifold $\mathbb Q_R=\big(\mathbb R^7,Q,(\cdot,\cdot)_R\big)$ with the sub-Riemannian structure $\big(Q,(\cdot,\cdot)_R\big)$. 

Geodesics (or normal extremals) in the sub-Riemannian geometry are defined as a projection of solutions of the Hamiltonian equations on the underlying manifold. Let us present the corresponding Hamiltonian system. We denote by $T_p\mathbb Q_R$, $T_p^*\mathbb Q_R$ the tangent and cotangent space at $p\in\mathbb Q_R$ respectively and by $T\mathbb Q_R$, $T^*\mathbb Q_R$ the corresponding tangent and cotangent bundle. Thus, if $(p,\lambda)\in T^*\mathbb Q_R$ then the restriction of $\lambda$ on the subspace $Q(p)$ of $T_p\mathbb Q_R$ is well defined and making use of the inner product we define a Hamiltonian function on $T_p^*\mathbb Q_R$ by $$H_R(p,\lambda)=\frac{1}{2}\sum\limits_{i=0}^3\langle\lambda,X_i\rangle^2,$$ where by $\langle \cdot,\cdot\rangle$ we denoted the pairing between vector spaces $T_p\mathbb Q_R$ and $T_p^*\mathbb Q_R$. This definition coincides with the definition of the norm of the linear functional $\lambda$ over the vector space $Q(p)$. If we write $p=(x_0,\ldots,x_3,z_1,z_2,z_3)$ and $\lambda=(\xi_0,\ldots,\xi_3,\theta_1,\theta_2,\theta_3)$ then the Hamiltonian function can be rewritten in the following form
$$H_R(p,\lambda)=\frac{1}{2}|\xi|^2+\frac{1}{8}|x|^2|\theta|^2+\xi^T\Omega x.$$
We get the corresponding Hamiltonian system 
\begin{equation}\label{2}
\begin{split} \left\{\array{lll}\dot x &= \frac{\partial H_R}{\partial
\xi}=\xi+\frac{1}{2}\Omega x
\\
\dot z_{k} &=  \frac{\partial H_R}{\partial
\theta_{k}}=\frac{1}{4}|x|^2\theta_{k}+\frac{1}{2}\xi^TI_kx,\quad k=1,2,3
\\
\dot \xi &= -\frac{\partial H_R}{\partial
x}=-\frac{1}{4}|\theta|^2x+\frac{1}{2}\Omega\xi
\\
\dot \theta &= -\frac{\partial H_R}{\partial z}=0.
\endarray\right.
\end{split} \end{equation} Here $|\theta|=(\sum_{k=1}^{3}\theta_k^2)^{1/2}$, $|x|=(\sum_{i=0}^{3}x_i^2)^{1/2}$. After the simplification, we obtain that $\theta_k$ are constant and 
\begin{equation}\label{eq:RH4}
 \ddot x=\Omega \dot x, \quad x\in\mathbb R^4
\end{equation}
\begin{equation}\label{eq:RH3}
 \dot z_k=\dot x^T I_kx, \quad k=1,2,3.
\end{equation}

The solution of these equations and detailed calculations can be found in~\cite{CM}.

\subsection{Sub-Lorentzian manifold $\mathbb Q_L$}

\medskip

Let us change the positively definite metric $(\cdot,\cdot)_R$ on $Q$ on the Lorentz metric (that is nondegenerate metric of index $1$) $(\cdot,\cdot)_L$ such that \begin{equation}\label{metric}(X_i,X_j)_L=\delta_{ij},\quad (X_0,X_0)_L=-1,\quad (X_i,X_i)_L=1,\ i=1,2,3.\end{equation}
We call the triple $\mathbb{Q}_L=(\mathbb{R}^7, Q,(\cdot,\cdot)_L)$ the {\it sub-Lorentzian manifold} or the {\it sub-Lorentzian $H$-type group} and the pair $(Q,(\cdot,\cdot)_L)$ is named by the {\it sub-Lorentzian structure} on $\mathbb{R}^7$. 

We define the casual character on $\mathbb{Q}_L$. Fix a point $p\in \mathbb R^7$. A horizontal vector $v\in Q_p$ is called {\it timelike} if $(v,v)_L<0$, {\it spacelike} if $(v,v)_L>0$ or $v=0$, null if $(v,v)_L=0$ and $v\neq0$, nonspacelike if $(v,v)_L\leqslant 0$. A horizontal curve is called timelike if its tangent vector is timelike at each point. Spacelike, null and nonspacelike curves are defined similarly. The choice of the sub-Lorentzian metric~\ref{metric} implies that the horizontal vector field $X_0$ is timelike and other horizontal vector fields $X_i$, $i=1,2,3$, are spacelike. We call $X_0$ the {\it time orientation} on $\mathbb{Q}_L$. Then a nonspacelike vector $v\in Q_p$ is called future directed if $(v,X_0(p))_L<0$, and it is called past directed if $(v,X_0(p))_L>0$. Throughout this paper f.d. stands for "future directed", t. for "timelike", and nspc. for "nonspacelike".

We would like to start the description of $\mathbb{Q}_L$ with finding geodesics, that is by definition, the projections of a solution of the associated Hamiltonian system on $\mathbb{Q}_L$. 

We construct a Hamiltonian system with respect to sub-Lorentzian metric. Locally the Hamiltonian function associated with the Lorentzian metric can be defined in the following way:
\begin{equation*}
    H_L(p,\lambda)=-\frac{1}{2}\langle\lambda, X_0\rangle^2+\frac{1}{2}\sum\limits_{i=1}^3\langle\lambda,X_i\rangle^2.
\end{equation*} If we use the coordinates for $(p,\lambda)$ as in the previous subsection, then the Hamiltonian becomes
\begin{gather*}
   H_L(\xi,\theta,x,z)=\frac{1}{2}\xi^T\eta\xi+\frac{1}{8}(\Omega x)^T\eta\Omega x+\xi^T\eta\Omega x.\notag
\end{gather*} Here $\Omega(\theta)$ is the matrix given by~\eqref{eq:omega} and $\eta$ is the matrix of the Minkowskii metric tensor.\begin{equation}\label{eq:eta}
    \eta=\left(\begin{matrix}
          -1& 0 &0 &0\\
           0& 1 &0 &0\\
            0& 0 &1 &0\\
            0 &0 &0 &1\\
     \end{matrix}\right).
\end{equation}

The corresponding Hamiltonian system takes the form
\begin{equation}\label{eq:hamsys}\begin{split}
     &\dot{x}=\frac{\partial H_L}{\partial \xi}=\eta\xi+\frac{1}{2}Ax,\\
     &\dot{z}_k=\dfrac{\partial H_L}{\partial \theta_k}=\frac{1}{4}(I_kx)^T\eta \Omega x+\frac{1}{2}\xi^T\eta I_kx,\quad k=1,2,3,\\
     &\dot{\xi}=-\frac{\partial H_L}{\partial x}=\frac{1}{4}\eta A^2x-\frac{1}{2}A^T\xi,\\
     &\dot{\theta}=-\frac{\partial H_L}{\partial z}=0,\end{split}
\end{equation}  where the paricipating matrix $A=\eta\Omega$ is a constant matrix of the parameters and
\begin{gather}\label{A2}
    A^2=\left(%
   \begin{array}{cccc}
      |\theta|^2&0&0&0\\
       0&\theta_1^2-\theta_2^2-\theta_3^2&-2\theta_1\theta_3&-2\theta_1\theta_2\\
       0&-2\theta_1\theta_3&-\theta_1^2-\theta_2^2+\theta_3^2&2\theta_2\theta_3\\
       0&-2\theta_1\theta_2&2\theta_2\theta_3&-\theta_1^2+\theta_2^2-\theta_3^2
   \end{array}%
    \right).
\end{gather} Here by symbol $ |\theta|$ is denoted the expression $\sqrt{\theta_1^2+\theta_2^2+\theta_3^2}$. Notice that $A^T=-\Omega\eta$.

After not intricate calculations two first equations of system~\eqref{eq:hamsys} roll up to the following linear system of ordinary differential equations
\begin{equation}\label{sys:1}
       \ddot{x}=A
       \dot{x},
\end{equation}
\begin{equation}\label{eq:RH31}
 \dot z_k=\dot x^T I_kx, \quad k=1,2,3,
\end{equation} that gives the equations for geodesics on~$\mathbb{Q}_L$. Here the conditions \eqref{eq:RH31} are derived from the second line of the system \eqref{eq:hamsys} by substituting $\xi$ from the first line of this system.
 The exact formulae for obtained geodesics see in \cite{KM}.

\section{Electromagnetic fields}\label{sec:3}

In this section we briefly introduce the notion of an electromagnetic field in order to explain the relation between the motion of the charged particle in an electromagnetic field and a sub-Lorentzian geodesic.

Consider the Minkowski spacetime $M$ with the Lorentzian metric $\eta$ in it. Let $(m,\alpha)$ be a charged particle in $M$ of a unite charge and a constant mass $m$ with a trajectory $\alpha$. Charged particles create an electromagnetic field and also respond on the fields created by other particles. An electromagnetic field in $M$ can be described by using two 3-dimensional vectors $\overrightarrow{E}$ and $\overrightarrow{B}$ that express electric and magnetic components respectively. Electromagnetic fileds in the space free of charge satisfies to four Maxwell's equations 
\begin{equation}\label{eq:max}
\begin{array}{ccc}
&\nabla\cdot\overrightarrow{B} =0\qquad & \nabla\cdot \overrightarrow{E}=0\\
& \nabla\times \overrightarrow{E}+\frac{\partial \overrightarrow{B}}{\partial x_0}=0,\qquad & \nabla\times \overrightarrow{B}-\frac{\partial \overrightarrow{E}}{\partial x_0}=0,\end{array}
\end{equation} where $x_0$ stands for the time coordinate in $M$ and permittivity and permeability are supposed to be constant and equal to 1. 
If we use the covariant formulation then Maxwell's equations can be written in the nice symmetric form \begin{equation}\label{eq:max1} dF=0,\qquad d*F=0,\end{equation} where $F$ is a 2-form field in 4-dimensional spacetime corresponding to antisymmetric electromagnetic tensor field \begin{equation}\label{eq:F}
 F=F_{\alpha\beta}=\left(%
   \begin{array}{cccc}
       0& -E_3 & -E_2 & -E_1 \\
      E_3 & 0 &B_1&-B_2\\
      E_2&-B_1&0&B_3\\
      E_1&B_2&-B_3&0
   \end{array}%
    \right) 
\end{equation} The operator $d$ is the exterior derivative, a coordinate and metric independent differential operator, and $*$ is the Hodge star operator that is linear transformation from the space of $2$ form into the space of two-forms defined by the metric in Minkowski space, see for instance~\cite{KN}. 

While Maxwell equations describes how electrically charged particles and objects give rise to electric and magnetic fields, the Lorentz force law completes that picture by describing the force acting on a moving charged particle in the presence of electromagnetic fields, see, for instance~\cite{Naber}. This effect is described by Lorentz equation  $$\frac{dP}{dt}=e\eta FU=F^{\mu\nu}U_{\beta},$$ where $U$ is the particle's world velocity, $P=mU$ its world momentum, $t$ is the proper time of the particle, and $F$ is an electromagnetic tenzor field. 

Let $(e_0,e_1,e_2,e_3)$ be any admissible basis in $M$; that is orthonormal, $e_0$ responds for the time coordinate and $(e_1,e_2,e_3)$ for space coordinates. As it was mentioned at each point of $M$ the linear transformation $F$ can be defined in terms of the classical electric and magnetic $3$-vectors $\overrightarrow{E}=E_1e_1+E_2e_2+E_3e_3$ and $\overrightarrow{B}=B_1e_1+B_2e_2+B_3e_3$ at that point. Set
\begin{equation}\label{eq:A}
 A=\eta F=\left(%
   \begin{array}{cccc}
       0& E_3 & E_2 & E_1 \\
      E_3 & 0 &B_1&-B_2\\
      E_2&-B_1&0&B_3\\
      E_1&B_2&-B_3&0
   \end{array}%
    \right).
\end{equation} The transformation $A$ which influences on the charged particle is often called Lorentz force. In order to find eigenspaces of $A$ that are invariant subspaces of this linear transformation we come to the characteristic equation. 
\begin{equation*}
    \det(A-\lambda I)=\lambda^4+(|\overrightarrow{B}|^2-|\overrightarrow{E}|^2)\lambda^2-(\overrightarrow{E}\cdot\overrightarrow{B})^2
     =0,
\end{equation*} where $|\overrightarrow{E}|^2=(E_1)^2+(E_2)^2+(E_3)^2$ and $\overrightarrow{E}\cdot \overrightarrow{B}=E_1B_1+E_2B_2+E_3B_3$. 
The algebraic combinations $|\overrightarrow{B}|^2-|\overrightarrow{E}|^2$ and $\overrightarrow{E}\cdot\overrightarrow{B}$ are the same in all admissible frames and are called Lorentz invariants. If both of them are equal to zero (i.e., $\overrightarrow{E}$ and $\overrightarrow{B}$ are perpendicular and have the same magnitudes): $|\overrightarrow{B}|^2-|\overrightarrow{E}|^2=\overrightarrow{E}\cdot\overrightarrow{B}=0$, then $A$ is called {\it null} transformation, otherwise $A$ is said to be {\it regular}. 

Every regular skew-symmetric with respect to the Lorentz metric linear transformation $A\colon M\to M$ has a 2-dimensional invariant subspace $V$ such that $V\bigcap V^{\bot}=\{0\}$.
There exist a basis, that is called the canonical basis, relative to which the matrix of regular skew-symmetric linear tranformation $A$ has the form 
\begin{equation}\label{eq:Fcan}
 \left(%
   \begin{array}{cccc}
       0 & \varepsilon&0 & 0\\
       \varepsilon &0&0& 0\\
       0& 0 & 0& -\delta  \\
       0& 0 & \delta  & 0 \\
   \end{array}%
    \right),
\end{equation} where $\delta$ and $\varepsilon$ are nonnegative real values such that $|\overrightarrow{B}|^2-|\overrightarrow{E}|^2=\delta^2-\varepsilon^2$, $\overrightarrow{E}\cdot\overrightarrow{B}=\delta\varepsilon$. Now the eigenvalues of $A$ are easy to calculate since the characteristic equation becomes $\lambda^4+(\delta^2-\varepsilon^2)\lambda^2-\delta^2\varepsilon^2=0$, i.~e., $(\lambda^2-\varepsilon^2)(\lambda^2+\delta^2)=0$, which has the following solutions: $\lambda_{1,2}=\pm\varepsilon$ and $\lambda_{3,4}=\pm i\delta$. 

\begin{definition}
    The linear transformation $T\colon M\to M$ defined by $$T=\frac{1}{4\pi}\left[\frac{1}{4}\tr(A^2)I-A^2\right],$$ where $A^2=A\circ A$, I is the identity transformation $I(x)=x$ for every $x\in M$ and $\tr(A^2)$ is the trace of $A^2$, is called the {\it energy-momentum transformation} associated with $A$.
\end{definition}
Observe that $T$ is symmetric with respect to the Lorentzian inner product and is trace-free, i.~e. $\tr T=0$. 
The term $-T_1^1=Te_0\cdot e_0=\frac{1}{8\pi}(|\overrightarrow{E}|^2+|\overrightarrow{B}|^2)$ is called the {\it energy density} in the given frame of reference for the electromagnetic field $F$ with the form~ \eqref{eq:F}. The 3-vector $\frac{1}{4\pi}\overrightarrow{E}\times\overrightarrow{B}=(E_2B_3-E_3B_2)e_1+(E_3B_1-E_1B_3)e_2+(E_1B_2-E_2B_1)e_3$ is called the {\it Poynting 3-vector} and describes the energy flux of the field. Finally, the $3\times 3$ matrix $(T^i_j)_{i,j=1,2,3}$ is known as {\it Maxwell stress tensor} of the field in the given frame. Thus, the content of the matrix of $T$ determines the energy content of the field $F$ in the corresponding basis.

\section{Sub-Lorentzian geodesics and the trajectories of the particles}\label{sec:4}

The noncommutative multiplication law~\eqref{eq:hprod} of quaternion $\mathbb H$-type group $(\mathbb R^7,\circ)$ defines the 
non-integrable distribution $Q=\spn\{X_0,\ldots,X_3\}$ or in the covariant language the nonholonomic constraints 
$\omega_1=\omega_2=\omega_3=0$. The curvature of the distribution gives rise to the skew symmetric transformation 
$\Omega$ in 4-dimensional space~\eqref{eq:omega}. 
Independently whether this space has Euclidean structure or it is the Minkowskii space the antisymmetric 2-form $\Omega$ 
defines the electromagnetic field, since it trivially satisfies the Maxwell equations~\eqref{eq:max1}. We emphasize that 
the geometry of the nonholonomic manifold $(\mathbb R^7,\circ)$ is related to the geometry of 4-dimensional space where 
a constant electromagnetic field acts. Given a positively definite metric $(\cdot,\cdot)_R$ (Riemannian metric) on the 
nonholonomic distribution $Q$ we obtain a sub-Riemannian manifold $\mathbb Q_R$. The Hamiltonian function $H_R$ in this 
case is reduced to the Lorentz equation in the Euclidean space $\mathbb R^4$ given by~\eqref{eq:RH4}. Equation~\eqref{eq:RH4} 
describes a motion of charged particle of unit charge in magnetic field $\Omega$ in 4-dimentional Euclidean space. 
Replacement of a positively definite metric $(\cdot,\cdot)_R$ on a nondegenerate indefinite metric, the Lorentzian metric 
$(\cdot,\cdot)_L$, leads to the relativistic Lorentz equation~\eqref{sys:1} in the Minkowskii space. 
It has more connections to physics since it is related to the motion of the charged particle in electromagnetic 
field in the Minkowskii space that is closely connected to the general relativity.

Our aim is to find geodesics that are projections on $\mathcal R^7$ of the solutions to the Hamiltonian system for $H_L$ defined on $\mathbb Q$. The Hamiltonian system is reduced to the equations~\eqref{sys:1} and ~\eqref{eq:RH31}.  Equations~\eqref{sys:1} are ordinary differential equations in 4-dimensional space with $A$ that is skew-symmetric with respect to the Lorentzian metric $\eta$. This makes us to endow the 4-dimensional space with the Lorentzian metric, producing the Minkowskii space $\mathbb M=(\mathbb R^4,\eta)$. Therefore, at each point of $\mathbb M$ we can elect a model of electromagnetic field which corresponds to a linear skew-symmetric transformation $A\colon \mathbb M\to \mathbb M$ that assigns to the world velocity $\dot{x}=U$ of a charged particle passing through that point the change in world momentum $\frac{dP}{dt}$ that the particle should expect due to the field. Since we can assume that the charge and the mass of the particle equals 1 we get $P=U$ and the Lorentz force law becomes $\frac{dU}{dt}=AU$ that is exactly equations~\eqref{sys:1}. If we set $E_1=-B_1=\theta_2$, $E_2=-B_2=\theta_3$, $E_3=-B_3=\theta_1$, then we conclude that equation~\eqref{sys:1} describes the motion of a particle of unit charge in the constant electromagnetic field with $\overrightarrow{E}=-\overrightarrow{B}$. In this case one of the Lorentz invariant is zero: $|\overrightarrow{B}|^2-|\overrightarrow{E}|^2=0$, and the other $\overrightarrow{E}\cdot\overrightarrow{B}$ is equal to $-|\theta|^2$, where $|\theta|^2=\theta_1^2+\theta_2^2+\theta_3^2$. The matrix $A^2$ given by~\eqref{A2} has zero trace. Therefore, the energy momentum transformation is given up to a constant multiplied by symmetric matrix $A^2$. The energy density is equal to $-2|\theta|^2$. The Poynting 3-vector in our case is zero vector. The Maxwell stress tensor is given by  
\begin{gather}\label{Mst}
    \left(%
   \begin{array}{ccc}
      \theta_1^2-\theta_2^2-\theta_3^2&-2\theta_1\theta_3&-2\theta_1\theta_2\\
      -2\theta_1\theta_3&-\theta_1^2-\theta_2^2+\theta_3^2&2\theta_2\theta_3\\
      -2\theta_1\theta_2&2\theta_2\theta_3&-\theta_1^2+\theta_2^2-\theta_3^2
   \end{array}%
    \right).
\end{gather}

Let us find a canonical basis for the matrix $A$. First we find a plane in $\mathbb{R}^3$ in which both of the vectors $\overrightarrow{E}$ and $\overrightarrow{B}$ are lying, then we rotate it in such a way that it will coincide with a plane $x_3=0$ in $\mathbb{R}^3$. That is, let us choose a right-handed orthonormal basis $\{\widehat{e}_1,\widehat{e}_2,\widehat{e}_3\}$ of the space $\mathbb{R}^3=\spn(e_1,e_2,e_3)$ in which $\widehat{E}_3=\widehat{B}_3=0$. There are infinitely many planes containing both vectors $\overrightarrow{E}$ and $-\overrightarrow{E}$. Fix one of these planes: for example, the one passing through the axis $x_2=(0,1,0)$: 
\begin{equation}\label{eq:2}-x_1E_3+x_3E_1=0.
 \end{equation}
 With the help of the rotation 
\begin{equation}\label{eq:R}
 R=\left(%
   \begin{array}{ccc}
       \cos\alpha& 0 & \sin\alpha \\
       0  & 1 & 0\\
       -\sin\alpha & 0 &\cos\alpha 
   \end{array}%
    \right)
\end{equation} we turn it on the angle $\alpha=\arccos \frac{E_1}{\sqrt{E_1^2+E^2_3}}$ between the plane \eqref{eq:2} and $x_3=0$ around the axis $x_2$ so that the third coordinate of $\overrightarrow{E}$ equals to zero. Consider now a Lorentz transformation 
\begin{equation*}
 R_1=\left(%
   \begin{array}{cc}
       1& 0  \\
       0  & R
   \end{array}%
    \right)
\end{equation*} of a basis $\{e_0,e_1,e_2,e_3\}$, where $R$ is given by~\eqref{eq:R}. It yields a new admissible coordinate system $\{\widehat{e}_0,\widehat{e}_1,\widehat{e_2},\widehat{e}_3\}$ in which $\overrightarrow{\widehat{E}}=(\sqrt{E_1^2+E_3^2}, E_2,0)$ and $\overrightarrow{\widehat{B}}=(-\sqrt{E_1^2+E_3^2}, -E_2,0)$ and matrix $\widehat{A}$ can be defined in the following way:
\begin{equation*}
 \widehat{A}=\left(%
   \begin{array}{cccc}
       0& 0 & \widehat{E}_2 & \widehat{E}_1 \\
       0  & 0 & -\widehat{E}_1 & \widehat{E}_2\\
       \widehat{E}_2 & \widehat{E}_1 & 0 & 0\\
       \widehat{E}_1&-\widehat{E}_2&0&0
   \end{array}%
    \right),
\end{equation*} where $\widehat{E}_1=\sqrt{E_1^2+E_3^2}=-\widehat{B}_1=\sqrt{\theta_1^2+\theta_2^2}$, $\widehat{E}_2=E_2=-\widehat{B}_2=\theta_3$, and $\widehat{E}_3=-\widehat{B}_3=0$.

Since $A:\mathbb M\to \mathbb M$ is regular skew-symmetric transformation, then it has a 2-dimensional invariant subspace $U$ such that $U\cap U^{\bot}=\{0\}$. Then $U^{\bot}$ is also a 2-dimensional invariant subspace for $A$ and there exist real numbers $\varepsilon\geqslant0$ and $\delta\geqslant0$ such that \cite{Naber}
\begin{equation*}\begin{split}
                     &A^2u=\varepsilon^2 u\quad\mbox{for all}\quad u=(u_1,u_2)\in U\quad\mbox{and}\\
                     &A^2v=-\delta^2 v\quad\mbox{for all}\quad v=(v_1,v_2)\in U^{\bot}.
                 \end{split} 
\end{equation*}

Take any future directed unit timelike vector $\tilde e_0$ in $U$, for example $(1,0,0,0)=\frac{v_1+v_2}{2|\theta|}$, where $v_i$ are eigenvectors of matrix $A$ (see \cite{KM}). Then spacelike unit vector $\tilde e_1$ in $U$ and the real value $\varepsilon\geqslant0$ can be found from the conditions $A\tilde e_0=\varepsilon \tilde e_1$ and $A\tilde e_1=\varepsilon \tilde e_0$. We get $\tilde e_1=(0,-\frac{\theta_1}{|\theta|},\frac{\theta_3}{|\theta|},\frac{\theta_2}{|\theta|})$ and $\varepsilon=|\theta|$. Now, let $\tilde e_2$ be an arbitrary unit spacelike vector in $U^{\bot}$, for example, select  $\tilde e_2=\frac{v_3+v_4}{2|\theta|\sqrt{\theta_1^2+\theta_2^2}}=\frac{1}{|\theta|\sqrt{\theta_1^2+\theta_2^2}}(0,\theta_1\theta_3,\theta_1^2+\theta_2^2,-\theta_2\theta_3)$. Then construct $\tilde e_3$ satisfying $A\tilde e_2=\delta \tilde e_3$ and $A\tilde e_3=-\delta \tilde e_2$. We obtain $\tilde e_3=(0,-\frac{\theta_2}{\sqrt{\theta_1^2+\theta_2^2}},0,-\frac{\theta_1}{\sqrt{\theta_1^2+\theta_2^2}})$ and $\delta =|\theta|$. Thus, $\{\tilde e_0,\tilde e_1,\tilde e_2,\tilde e_3\}$ is an orthnormal basis for $\mathbb M$, which is called {\it canonical} and 
\begin{equation}\label{eq:Acanonical}
 \tilde A=\left(%
   \begin{array}{cccc}
       0& |\theta| & 0 & 0 \\
       |\theta|  & 0 & 0 & 0\\
       0 & 0 & 0 & -|\theta|\\
       0& 0& |\theta| &0
   \end{array}%
    \right).
\end{equation} Electric and magnetic fields corresponding to this transformation are $\overrightarrow{E}=\varepsilon \tilde e_3=|\theta|\tilde e_3$ and $\overrightarrow{B}=\delta \tilde e_3=|\theta|\tilde e_3$, so that an observer in this frame will measure them in the same direction (of $x_3$-axis) and of magnitude $|\theta|$.
Matrix $\widetilde A$ is a block-type matrix, where $(2\times2)$-block in the left upper corner coincides with a matrix in sub-Lorentzian Heisenberg case and the right lower $(2\times2)$-block coincides with a matrix in usual sub-Riemannian Heisenberg case (see \cite{KM}).

Let us denote by $P$ the matrix which columns are orthnormal basis vectors $\tilde e_0,\ldots,\tilde e_3$. Then in new basis $\tilde A=P^{-1}AP$ and the vector $x\in \mathbb M$ is of the form $x=Px_{old}$, where $x_{old}$ is a vector in old coordinates. Therefore, $x_{old}=P^{-1}x$.

It is clear that real eigenvalues of $A$ are $\lambda=\pm |\theta|$ from the canonical form \eqref{eq:Acanonical} . Therefore, eigenspaces are $\spn\{\tilde e_3+ \tilde e_4\}$ and $\spn\{\tilde e_3-\tilde e_4\}$ respectively. Directions $\tilde e_3\pm \tilde e_4$ are called {\it principal null directions} of $A$.

The Lorentz equation~\eqref{sys:1} in canonical coordinates takes the form
\begin{equation*}
 \left(%
   \begin{array}{c}
       \frac{d\dot x_0}{dt}\vspace{3pt}\\
       \frac{d\dot x_1}{dt}\vspace{3pt}\\
       \frac{d\dot x_2}{dt}\vspace{3pt}\\
       \frac{d\dot x_3}{dt}
   \end{array}%
    \right)=\left(%
   \begin{array}{cccc}
       0& |\theta| & 0 & 0 \\
       |\theta|  & 0 & 0 & 0\\
       0 & 0 & 0 & -|\theta|\\
       0& 0& |\theta| &0
   \end{array}%
    \right)  \left(%
 \begin{array}{c}
       \dot x_0\\ 
       \dot x_1\\
       \dot x_2\\
       \dot x_3
   \end{array}%
    \right)=\left(%
 \begin{array}{c}
       \;|\theta| \dot x_1\\ 
       \;|\theta| \dot x_0\\
       \;\!\!\!\!\!\!\!\!{}-|\theta| \dot x_3\\
       \;|\theta| \dot x_2
   \end{array}%
    \right).
\end{equation*}

It splits into 2 independent systems
\begin{gather}
     \left\{%
   \begin{array}{c}
       \frac{d\dot x_0}{dt}=|\theta| \dot x_1,\vspace{3pt}\\
       \frac{d\dot x_1}{dt}=|\theta| \dot x_0
   \end{array}%
    \right. \quad\qquad 
     \left\{%
   \begin{array}{c}
       \!\!{}\frac{d\dot x_2}{dt}=-|\theta| \dot x_3,\vspace{3pt}\\
       \frac{d\dot x_3}{dt}=|\theta| \dot x_2
   \end{array}%
    \right. .\notag
\end{gather} We wish to solve this system under the initial conditions $x(0)=0$, $\dot x(0)=\dot x_0$. The solution for $\dot x(t)$ is
\begin{equation}\label{eq:U}
\begin{split}
    &\dot x_0(t)=\dot x_0(0)\cosh (|\theta| t)+\dot x_1(0)\sinh(|\theta| t),\\
    &\dot x_1(t)=\dot x_0(0)\sinh (|\theta| t)+\dot x_1(0)\cosh(|\theta| t),\\
    &\dot x_2(t)=\dot x_2(0)\cos (|\theta| t)-\dot x_3(0)\sin(|\theta| t),\\
    &\dot x_3(t)=\dot x_2(0)\sin (|\theta| t)+\dot x_3(0)\cos(|\theta| t).\\
\end{split} 
\end{equation} 
Integrating these expressions, we get the solution $x(t)$, that we write in the matrix form as $x(t)=W\dot{x}_0$ with
\begin{equation}\label{eq:W}
     W=\frac{1}{|\theta|}\left(%
 \begin{array}{cccc}
       \sinh(|\theta| t)&\cosh(|\theta| t)-1&0&0\\
       \cosh(|\theta| t)-1&\sinh(|\theta| t)&0&0\\
       0&0&\sin(|\theta| t)&\cos(|\theta| t)-1\\
       0&0&1-\cos(|\theta| t)&\sin(|\theta| t).
   \end{array}%
    \right).
\end{equation} 
Now, let us introduce following notation: $v_1=-\dot x_0^2(0)+\dot x_1^2(0)$, $v_2=\dot x_2^2(0)+\dot x_3^2(0)$.
\begin{lemma}\label{hyperbola}
   The projection of the geodesic onto $(x_0,x_1)$-plane is a brunch of the hyperbola with the canonical equation
\begin{equation}\label{eqhyperbola}\left(x_0+\frac{\dot x_1(0)}{|\theta|}\right)^2-\left(x_1+\frac{\dot x_0(0)}{|\theta|}\right)^2=\frac{v_1}{|\theta|^2}.\end{equation}
\end{lemma}
\begin{proof}
     Since $$x_0(t)=\frac{1}{|\theta|}\big(\dot x_0(0)\sinh(|\theta|t)+\dot x_1(0)(\cosh(|\theta|t)-1)\big)$$
$$x_1(t)=\frac{1}{|\theta|}\big(\dot x_0(0)(\cosh|\theta|t-1)+\dot x_1(0)\sinh(|\theta|t)\big),$$ we calculate that $-x_0^2(t)+x_1^2(t)=\frac{2(-\dot{x}_0^2(0)+\dot{x}_1^2(0))}{|\theta|^2}\cdot\sinh^2\frac{|\theta|t}{2}$. This expression can be rewritten as stated in~\eqref{eqhyperbola}
\end{proof}

\begin{lemma}\label{circle}
   The projection of the geodesic onto $(x_2,x_3)$-plane is a circle with the center at $\left(-\frac{\dot x_3^0}{|\theta|},\frac{\dot x_2^0}{|\theta|}\right)$ of the radius $\frac{\sqrt{v_2}}{|\theta|}$.
\end{lemma}
\begin{proof}
     Since $$x_2(t)=\frac{1}{|\theta|}\big(\dot x_2(0)\sin|\theta|t+\dot x_3(0)(\cos|\theta|t-1)\big)$$  $$ x_3(t)=\frac{1}{|\theta|}\big(\dot x_2(0)(1-\cos|\theta|t)+\dot x_3(0)\sin(|\theta|t)\big),$$ we get $x_2^2(t)+x_3^2(t)=\frac{4v_2}{|\theta|^2}\cdot\sinh^2\frac{|\theta|t}{2}$. The expression leads to
\begin{equation}\label{eq:circle}
   \left(x_2+\frac{\dot x_3(0)}{|\theta|}\right)^2+\left(x_3-\frac{\dot x_2(0)}{|\theta|}\right)^2=\frac{v_2}{|\theta|^2}.   
\end{equation}
\end{proof}

The horizontality conditions \eqref{eq:RH3} in canonical basis have the form 
$$\dot{z}_k=(I_kx_{old},\dot{x}_{old})=(I_kP^{-1}x,P^{-1}\dot{x})=(PI_kP^{-1}x,\dot{x}),$$ 
where the matrices
$$J_1=PI_1P^{-1}=-\frac{1}{|\theta|}\Omega(\theta),$$
$$J_2=PI_2P^{-1}=\dfrac{1}{\sqrt{\theta_1^2+\theta_2^2}}\left(\begin{array}{cccc}
       0&\theta_2&0&\theta_1\\
       -\theta_2&0&\theta_1&0\\
       0&-\theta_1&0&\theta_2\\
       -\theta_1&0&-\theta_2&0
   \end{array}%
    \right),$$
$$J_3=PI_3P^{-1}=\dfrac{1}{|\theta|\sqrt{\theta_1^2+\theta_2^2}}\left(\begin{array}{cccc}
       0&-\theta_1\theta_3&-(\theta_1^2+\theta_2^2)&\theta_2\theta_3\\
       \theta_1\theta_3&0&\theta_2\theta_3&\theta_1^2+\theta_2^2\\
       \theta_1^2+\theta_2^2&-\theta_2\theta_3&0&-\theta_1\theta_3\\
       -\theta_2\theta_3&-(\theta_1^2+\theta_2^2)&\theta_1\theta_3&0
   \end{array}%
    \right)
$$ are skew-symmetric with respect to usual euclidean metric.
More explicitely,
$\dot{z}=\mathcal P\dot{\tilde z}$, where $\dot{z}$ are coordinates in the canonical basis, $\dot{\tilde z}=(\dot{\tilde z}_1,\dot{\tilde z}_2,\dot{\tilde z}_3)$
are auxiliary expressions $\dot{\tilde z}_k=(I_kx,\dot{x})$, $\dot x$ is given by~\eqref{eq:U} and 
\begin{equation}\label{mathcalP}
     \mathcal P=\left(%
 \begin{array}{ccc}
       -\frac{\theta_1}{|\theta|}&-\frac{\theta_2}{|\theta|}&-\frac{\theta_3}{|\theta|}\\
       \frac{\theta_2}{\sqrt{\theta_1^2+\theta_2^2}}&-\frac{\theta_1}{\sqrt{\theta_1^2+\theta_2^2}}&0\\
       -\frac{\theta_1\theta_3}{|\theta|\sqrt{\theta_1^2+\theta_2^2}}&-\frac{\theta_2\theta_3}{|\theta|\sqrt{\theta_1^2+\theta_2^2}}&\frac{\sqrt{\theta_1^2+\theta_2^2}}{|\theta|}.
   \end{array}%
    \right).
\end{equation} 
Notice that $\mathcal P$ is an orthogonal transformation in $\mathbb{R}^3$ while the matrix $P$ represents the orthogonal transformation $\mathbb R^4$.
It is more convenient for us to work with the expressions $\tilde z_1$, $\tilde z_2$, $\tilde z_3$. Then $z_1$, $z_2$, $z_3$ can be obtained by the orthogonal transformation $\mathcal P$ of $\tilde z_1$, $\tilde z_2$, $\tilde z_3$.
Taking into account that $x=W\dot x^0$ and \eqref{eq:U}, we calculate
\begin{equation*}
    \dot{\tilde z}_1(t)=\frac{1}{|\theta|}\big(v_1(\cosh(|\theta|t)-1)+v_2(\cos(|\theta|t)-1)\big),\\
\end{equation*}
\begin{equation*}\begin{split}
    &\dot{\tilde z}_2(t)= \frac{1}{|\theta|}\Big[2\cosh|\theta|t\cos|\theta|t\dot x_0(0)\dot x_2(0)+ 2\sinh|\theta|t\cos|\theta|t\dot x_1(0)\dot x_2(0) \\
    & \qquad\qquad\quad-2\sinh|\theta|t\sin|\theta|t\dot x_1(0)\dot x_3(0) 
    -2\cosh|\theta|t \sin|\theta|t\dot x_0(0)\dot x_3(0) \\
    &\qquad\qquad\quad+\cosh(|\theta|t)\big(\dot x_1(0)\dot x_3(0)-\dot x_0(0)\dot x_2(0)\big)+\sinh(|\theta|t)\big(\dot x_0(0)\dot x_3(0)-\dot x_1(0)\dot x_2(0)\big)
    \\&\qquad\qquad\quad-\cos(|\theta|t)\big(\dot x_0(0)\dot x_2(0)+\dot x_1(0)\dot x_3(0)\big)
     +\sin(|\theta|t)\big(\dot x_0(0)\dot x_3(0)-\dot x_1(0)\dot x_2(0)\big)\Big],
     \end{split}
\end{equation*}
      \begin{equation*}\begin{split}
    &\dot{\tilde z}_3(t)=\frac{1}{|\theta|}\Big[-2\cosh|\theta|t\cos|\theta|t\dot x_0(0)\dot x_3(0)- 2\sinh|\theta|t\cos|\theta|t\dot x_1(0)\dot x_3(0) \\
    &\qquad\qquad\quad-2\sinh|\theta|t\sin|\theta|t\dot x_1(0)\dot x_2(0) 
    -2\cosh|\theta|t \sin|\theta|t\dot x_0(0)\dot x_2(0)\\
    &\qquad\qquad\quad +\cosh(|\theta|t)\big(\dot x_0(0)\dot x_3(0)+\dot x_1(0)\dot x_2(0)\big)+\sinh(|\theta|t)\big(\dot x_0(0)\dot x_2(0)+\dot x_1(0)\dot x_3(0)\big)\\
    &\qquad\qquad\quad +\cos(|\theta|t)\big(\dot x_0(0)\dot x_3(0)-\dot x_1(0)\dot x_2(0)\big)
    +\sin(|\theta|t)\big(\dot x_0(0)\dot x_2(0)+\dot x_1(0)\dot x_3(0)\big)\Big].\end{split}
\end{equation*}
Let us use the following notation for the constants
\begin{gather*}
 a_1=\dot x_0(0)\dot x_3(0)+\dot x_1(0)\dot x_2(0),\quad
 a_3=\dot x_0(0)\dot x_2(0)+\dot x_1(0)\dot x_3(0),\\
 a_2=\dot x_0(0)\dot x_3(0)-\dot x_1(0)\dot x_2(0),\quad
 a_4=\dot x_0(0)\dot x_2(0)-\dot x_1(0)\dot x_3(0).
\end{gather*}
Thus, we integrate
$$\tilde z_1(t)=\frac{1}{|\theta|^2}\Big(v_1\big(\sinh(|\theta|t)-|\theta|t\big)+v_2\big(\sin(|\theta|t)-|\theta|t\big)\Big),$$
\begin{equation}\label{eq:z}
\begin{split}
&\tilde z_2(t)  =\frac{1}{|\theta|^2}\Big( a_1 \big(\cos(|\theta| t)\cosh(|\theta| t)-1\big)\\
&\qquad-a_2 \big(\sin(|\theta| t)\sinh(|\theta| t+\cos(|\theta| t)-\cosh(|\theta| t))\big)\\
&\qquad + a_3 \big(\cos(|\theta| t)\sinh(|\theta| t)-\sin(|\theta| t)\big)\\
&\qquad + a_4 \big(\sin(|\theta| t)\cosh(|\theta| t)-\sinh(|\theta| t)\big)\Big),\\
&\tilde z_3(t) =\frac{1}{|\theta|^2}\Big( a_4 \big(\cos(|\theta| t)\cosh(|\theta| t)-1\big)\\
&\qquad-a_3 \big(\sin(|\theta| t)\sinh(|\theta| t+\cos(|\theta| t)-\cosh(|\theta| t))\big)\\
&\qquad-a_2 \big(\cos(|\theta| t)\sinh(|\theta| t)-\sin(|\theta| t)\big)\\
&\qquad-a_1 \big(\sin(|\theta| t)\cosh(|\theta| t)-\sinh(|\theta| t)\big)\Big),
\end{split}
\end{equation}
Observe that 
\begin{gather*}
a_1^2+a_4^2=a_2^2+a_3^2 =  \big(\dot x_0^2(0)+\dot x_1^2(0)\big)v_2,\\
a_1a_2+a_3a_4 =  -v_1 v_2,\quad
a_1a_3-a_2a_4  =  2\dot x_0(0)\dot x_1(0)v_2.
\end{gather*}
Then the direct calculations yield 
\begin{equation}\label{eq:zsquare}
\begin{split}
 \tilde z_2^2(t)+\tilde z_3^2(t)  = \frac{4v_2}{|\theta|^4}\big(\cosh(|\theta| t)-\cos(|\theta| t)-\sin(|\theta| t)\sinh(|\theta| t)\big)\\
 \big((\dot x_0^2(0)+\dot x_1^2(0))\cosh(|\theta| t)+2\dot x_0(0)\dot x_1(0)\sinh(|\theta| t)-v_1\big).\end{split}\end{equation} 

\section{Reachable sets by geodesics}\label{sec:5}

We wish to describe the set of points in $\mathbb Q_L$ that can be reached from the origin $(x,z)=(0,0)$ by a geodesic: timelike, lightlike or spacelike. We can fix the starting point at the origin $O=(0,0)$, since the solutions of the Hamiltonian equations are invariant under the left translation. We start from the simple lemma, that is related to the case of null transformation $A$.

\begin{lemma}\label{lem:theta0}
If $|\theta|=0$, then the system~\eqref{eq:hamsys} with initial data $x(0)=z(0)=0$, $\xi(0)=\xi^0$, $\eta(0)=\eta_0$ has the solution $x(t)=\xi^0t$, $z(t)=0$, $\xi(t)=\xi^0$, $\eta(t)=\eta^0$. The projections to $(x,z)$-space are straight lines that are timelike if $|\xi^0|_L^2<0$, lightlike if $|\xi^0|_L^2=0$, and spacelike if $|\xi^0|_L^2>0$.
\end{lemma}

\begin{proof}
The condition $|\theta|=0$ immediately implies that $\xi(t)$ and $\eta(t)$ are constant and $x(t)=\xi^0t$.
Then $\dot{\tilde z}_k=\frac{t}{2}(\xi^0)^TI_k\xi^0=0$ since the matrix $I_k$ is skew symmetric.
\end{proof}

The description of the reachable set by causal curves is very complicated in general, therefore, we present here some particular cases. We mostly reduce our considerations to the sets that can be reached by geodesics, since the geodesics can not change its causal character. 

\subsection{Connectivity by geodesics between $(0,0)$ and $(x^1,z^1)$, where $|x^1|_L^2=0$}

We need to solve the equations $$\ddot x=\widetilde A\dot x,\qquad \dot{z}_k=\dot x^TJ_kx,\ \ k=1,2,3$$ with the boundary conditions
\begin{equation}\label{eq:boundary} x(0)=\tilde z(0)=0,\quad x(1)=x^1,\quad z(1)=z^1,\;\;\mbox{where}\;\;|x^1|_L^2=(\eta x^1,x^1)=0.\end{equation}

We find the relation between the initial velocity and the value $x(1)$. Substituting $t=1$ in $x(t)=W\dot{x}^0$, we obtain
\begin{equation}\label{eq:x1}
      |x^1|^2_L=\frac{4}{|\theta|^2}\left[\big(\!{}-\!{}\dot{x}_0^2(0)+\dot{x}_1^2(0)\big)\sinh^2\frac{|\theta|}{2}+\big(\dot{x}_2^2(0)+\dot{x}_3^2(0)\big)\sin^2\frac{|\theta|}{2}\right].
\end{equation} 
We also have $\dot x(0)=\dot{x}^0=W^{-1}(t)x(t)$, where
\begin{equation}\label{eq:Winverse}
     W^{-1}(t)=\frac{|\theta|}{2}\left(%
 \begin{array}{cccc}
       \frac{\sinh(|\theta| t)}{\cosh(|\theta| t)-1}&-1&0&0\\
       -1&\frac{\sinh(|\theta| t)}{\cosh(|\theta| t)-1}&0&0\\
       0&0&\frac{\sin(|\theta| t)}{1-\cos(|\theta| t)}&-1\\
       0&0&1&\frac{\sin(|\theta| t)}{1-\cos(|\theta| t)}
   \end{array}%
    \right).
\end{equation} from~\eqref{eq:W}. Putting $t=1$, we calculate
\begin{gather*}
     |\dot x^0|^2_L=(\eta \dot x^0,\dot x^0 )=\frac{|\theta|^2}{4}\left[\big(-x_0^2(1)+x_1^2(1)\big)\frac{1}{\sinh^2|\theta|}+\big(x_2^2(1)+x_3^2(1)\big)\frac{1}{\sin^2|\theta|}\right].
\end{gather*}

The expression \eqref{eq:x1} can be written in the form
\begin{equation}\label{eq:vyr}0=|x^1|^2_L=\frac{4}{|\theta|^2}\big(v_1\sinh^2\frac{|\theta|}{2}+v_2\sin^2\frac{|\theta|}{2}\big).\end{equation} 
Let us describe all possible initial velocities which lead to vanishing of the norm $|x^1|^2_L$.

\medskip

{\sc Case 1.} $|\theta|=0$. In this case Lemma~\ref{lem:theta0} implies $|x^1|_L^2=|\dot x^0|_L^2$ and we conclude that if $(x^1,z^1)$ is such that $|x^1|_L^2=0$ and $z^1=0$ then there is a unique geodesic, that is lightlike straight line, connecting $(x^1,z^1)$ with $(0,0)$. 

{\sc Case 2.} $v_2=0$, $|\theta|\neq 0$. Therefore, $v_1=0$. In this case  $|\dot{x}^0|_L^2=0$ and the possible geodesic is lightlike that remains lightlike for all $t\in[0,\pm\infty]$. Let us write the equations for these geodesics taking into account the condition on the initial velocity $\dot{x}_0(0)=\pm\dot{x}_1(0)$, $\dot{x}_2(0)=\dot{x}_3(0)=0$. 
\begin{equation}\label{eq:if}
     \mbox{If}\ \ \dot{x}_0(0)=\dot{x}_1(0),\ \ \mbox{then}\ \  x_0(t)=x_1(t)=\frac{\dot{x}_0(0)}{|\theta|}(\sinh|\theta|t+\cosh|\theta|t-1),\ \ x_2(t)=x_3(t)=0,\end{equation}
\begin{equation}\label{eq:if2}
   \mbox{if}\ \ \dot{x}_0(0)=-\dot{x}_1(0),\ \ \mbox{then}\ \  x_0(t)=-x_1(t)=\frac{\dot{x}_0(0)}{|\theta|}(\sinh|\theta|t-\cosh|\theta|t+1),\ \ x_2(t)=x_3(t)=0.
\end{equation} This is not obvious parametrization of straight lines $x(s)=(\dot x_0(0)s,\dot x_1(0)s,0,0)$, $\dot x_0^2(0)=\dot x_1^2(0)$.
In both cases $\dot{\tilde z}_{k}(t)=\dot x^T(t)I_kx(t)=0$, $k=1,2,3$ that implies $\tilde z_{k}(t)={z}_{k}(t)\equiv 0$. We conclude that the origin can be joined with a point $(x_0(1),x_1(1),0,0)$, where $x_0^2(1)=x_1^2(1)$ by lightlike geodesic that is straight line. Thus, case 2 is a particular situation of the case 1. 

Notice that the equations of the horizontal part~\eqref{eq:if} and~\eqref{eq:if2} of geodesics in $\mathbb Q_L$ coincides up to reparametrization with the equations of horizontal part for geodesics in $\mathbb{H}^1_L$ provided the initial velocity $\dot{x}_0(0)=\pm\dot{x}_1(0)$ (see \cite{KM}).

The results of cases 1 and 2 may be united in the following statement
\begin{theorem}
Let $A=(x^1,z^1)$ be a point such that $|x^1|^2_L=0$, $z_k=0$, $k=1,2,3$ and $x_0^2(1)+x_1^2(1)\neq 0$. Then there is a unique lightlike geodesic joining the origin and $A$ which is a straight line.
\end{theorem}

{\sc Case 3.}  $v_2\neq 0$, $|\theta|\neq 0$. We write the expression \eqref{eq:vyr} in the equivalent form
\begin{equation}\label{eq:ratio}
    \frac{\sin^2\frac{|\theta|}{2}}{\sinh^2\frac{|\theta|}{2}}=-\frac{v_1}{v_2}
\end{equation} and consider the following subcases.

3.1) Let $v_1=0$ then $|\dot{x}^0|^2=v_2>0$ and in this case all geodesics are spacelike. Then \eqref{eq:ratio} implies that $|\theta|=2\pi n$, $n\in \mathbb{N}\setminus \{0\}$.

3.1.1) Let us suppose that $\dot{x}_0(0)=\dot x_1(0)=0$ and $\dot x_2(0)$, $\dot x_3(0)$ are arbitrary. 
The main result is expressed in the following
\begin{theorem}\label{theorem2}
Given a point $A=(0,z^1)$ there are uncountably many spacelike geodesics connecting the origin with $A$. The geodesics are given by equations
$$x_0(t)=x_1(t)\equiv 0,\qquad x_2^2(t)+x_3^2(t)=\frac{2|z^1|}{\pi n}\sin^2(\pi n t))$$
$$z(t)=|z^1|\Big(\frac{\sin(2\pi n t)}{2\pi n}-t\Big)\big(z_1(1),z_2(1),z_3(1)\big).$$ The geodesics have the lengths $\sqrt{2\pi n|z^1|}$, $n\in \mathbb{N}\setminus \{0\}$.
\end{theorem}
\begin{proof}
We start from the equations for $x$-coordinates of geodesics. Since $x(t)=W\dot x^0$ we get 
\begin{equation}\label{x-geod1}\begin{array}{c}
       x_0(t)\\x_1(t)\\  x_2(t)\\ x_3(t)
\end{array}=\frac{1}{2\pi n}\left(\begin{array}{c}
       0\\0\\ \dot x_2(0)\sin(2\pi nt)+\dot x_3(0)(\cos(2\pi nt)-1)\\ \dot x_2(0) (1-\cos(2\pi nt))+\dot x_3(0)\sin(2\pi nt)
\end{array}\right)\end{equation} by~\eqref{eq:W}. The projection into $(x_2,x_3)$-plane is 
\begin{equation}\label{cardioid}
 x_2^2(t)+x_3^2(t)=4v_2\sin^2(\pi n t)),
\end{equation}
which by lemma \eqref{circle} can be rewritten as the canonical equation of the 
 circle on the $(x_2,x_3)$-plane  with the center at $\left(-\frac{\dot{x}_3(0)}{2\pi n},\frac{\dot{x}_2(0)}{2\pi n}\right)$ of the radius $\frac{\sqrt{v_2}}{2\pi n}$: \begin{equation}\left(x_2+\frac{\dot{x}_3(0)}{2\pi n}\right)^2+\left(x_3-\frac{\dot{x}_2(0)}{2\pi n}\right)^2=\frac{v_2}{(2\pi n)^2}.\end{equation} The number $n\in \mathbb{N}\setminus \{0\}$ reflects the number of turns along the circle for any fixed initial velocity $v_2$ and the radius of the circle. 
To exclude $v_2$ from the equations we use~\eqref{eq:z} and find
$$\tilde z_1(t)=\frac{v_2}{4\pi^2n^2}(\sin(2\pi n t)-2\pi nt),\quad    \tilde z_2(t)\equiv0,\quad \tilde z_3(t)\equiv0.$$
Set $t=1$ in the expression for $\tilde z_1(t)$ and get $v_2=-2\pi n\tilde z_1(1)$. Since $z(t)=\mathcal P \tilde z(t)$ with $\mathcal P$ given by~\eqref{mathcalP} we get \begin{equation}\label{tildez1} z^1=\tilde z_1(1)(\frac{-\theta_1}{|\theta|},\frac{\theta_2}{\sqrt{\theta_1^2+\theta_2^2}},-\frac{\theta_1\theta_3}{|\theta|\sqrt{\theta_1^2+\theta_2^2}}),\ \ |\theta|=2\pi n.\end{equation} It implies that $-\tilde z_1(1)=|z^1|$ and $v_2=2\pi n|z^1|$. Fixing $n\in \mathbb{N}\setminus \{0\}$, we fix the speed $v_2$ of a geodesic, but we still have the choice in the directions of $(\dot x_2(0),\dot x_3(0))$ that are parametrized by the unit circle. It gives uncountably many geodesics. 

We get the equations for $z$-coordinates from $z(t)=\mathcal P \tilde z(t)$ taking into account that the values of $\theta_k$ related to the values of $z^1$ by \eqref{tildez1}.

Since the geodesics are spacelike the length $l(\gamma)$ of a geodesic $\gamma$ can be calculated from the formula $$l(\gamma)=\int_0^1\sqrt{v_2}dt=\sqrt{2\pi n |z^1|}.$$
\end{proof}

\begin{remark}
Notice that in Lorentzian Heisenberg group $\mathbb H_L$ there is no geodesic of any causal type joining $(0,0,0)$ and $(0,0,z)$ with $z\neq0$~(see~\cite{KM}).
\end{remark}

3.1.2) We suppose now that $\dot x_0(0)=\dot x_1(0)\neq0$. We present first the auxiliary calculations and then formulate the main statement. We have $$a_1=a_3=\dot x_0(0)(\dot x_2(0)+\dot x_3(0)\big),\qquad -a_2=a_4=\dot x_0(0)(\dot x_2(0)-\dot x_3(0)\big).$$ Then the equations of geodesics take the form
\begin{equation}\label{x-geod2}\begin{array}{c}
       x_0(t)\\x_1(t)\\  x_2(t)\\ x_3(t)
\end{array}=\frac{1}{2\pi n}\left(\begin{array}{c}
       \dot x_0(0)(e^{2\pi nt}-1)\\\dot x_0(0)(e^{2\pi nt}-1)\\ \dot x_2(0)\sin(2\pi nt)+\dot x_3(0)(\cos(2\pi nt)-1)\\ \dot x_2(0) (1-\cos(2\pi nt))+\dot x_3(0)\sin(2\pi nt)
\end{array}\right)\end{equation}
The relations~\eqref{eq:z} lead to
\begin{gather}\label{eq:12}
\tilde z_1(t)=\frac{v_2}{(2\pi n)^2}(\sin(2\pi nt)-2\pi nt),\\
\tilde z_2(t)=\frac{\dot x_0(0)}{(2\pi n)^2}
\left[\dot x_2(0)\Big(\big(\cos(2\pi n t)-1\big)\big(e^{2\pi n t}+1\big)+\sin(2\pi n t)\big(e^{2\pi n t}-1)\big)\Big)\right.\nonumber \\
\qquad\;\;\;\quad\quad\left.+\dot x_3^0\Big(\big(\cos(2\pi n t)+1\big)\big(e^{2\pi n t}-1\big)-\sin(2\pi n t)\big(e^{2\pi n t}+1)\big)\Big)
\right],\nonumber\\
\tilde z_3(t)=\frac{\dot x_0(0)}{(2\pi n)^2}
\left[\dot x_2(0)\Big(\big(\cos(2\pi n t)+1\big)\big(e^{2\pi n t}-1\big)-\sin(2\pi n t)\big(e^{2\pi n t}+1)\big)\Big)\right.\nonumber \\
\qquad\;\;\;\quad\quad\left.-\dot x_3(0)\Big(\big(\cos(2\pi n t)-1\big)\big(e^{2\pi n t}+1\big)-\sin(2\pi n t)\big(e^{2\pi n t}-1)\big)\Big)
\right],\nonumber
\end{gather}
where $n\in\mathbb N\setminus\{0\}$. Setting $t=1$ in~\eqref{x-geod2} we get \begin{equation}\label{dotx0}\dot x_0(0)=x_0(1)\frac{2\pi n}{e^{2\pi n }-1}.\end{equation} The equations~\eqref{eq:12} imply \begin{equation}\label{eq:z1}v_2=-\tilde z_1(1)2\pi n,\qquad \frac{\dot x_0(0)\dot x_2(0)}{(2\pi n)^2}=\frac{\tilde z_3(1)}{2(e^{2\pi n}-1)},\qquad \frac{\dot x_0(0)\dot x_3(0)}{(2\pi n)^2}=\frac{\tilde z_2(1)}{2(e^{2\pi n}-1)}.\end{equation}

\begin{theorem}
Given a point $A=(x^1,z^1)$, $|x^1|_L=0$, $x_0(1)=x_1(1)\neq0$, $|z^1|\neq 0$, there are uncountably many spacelike geodesics connecting the origin $O=(0,0)$ with $A$. The equations of geodesics are given by expressions 
\begin{equation}
x_0(t)=x_1(t)=x_0(1)\frac{e^{2\pi nt}-1}{e^{2\pi n}-1},\quad x_2^2(t)+x_3^2(t)=-\frac{2\tilde z_1(1)}{\pi n}\sin^2(\pi nt))
\end{equation}
and $z(t)=\mathcal P\tilde z(t)$, where
\begin{gather*}
\tilde z_1(t)=\tilde z_1(1)(t-\frac{\sin(2\pi nt)}{2\pi n}),\\
\tilde z_2(t)=\frac{1}{2(e^{2\pi n }-1)}\left[\tilde z_3(1)\big((\cos(2\pi n t)-1)(e^{2\pi n t}+1)+\sin(2\pi n t)(e^{2\pi n t}-1)\big)\right.\nonumber \\
\left.\qquad\qquad\qquad\qquad\qquad +\tilde z_2(1)\big((\cos(2\pi n t)+1)(e^{2\pi n t}-1)-\sin(2\pi n t)(e^{2\pi n t}+1)\big)\right],\nonumber \\
\tilde z_3(t)=\frac{1}{2(e^{2\pi n }-1)}\left[\tilde z_3(1)\big((\cos(2\pi n t)+1)(e^{2\pi n t}-1)-\sin(2\pi n t)(e^{2\pi n t}+1)\big)\right.\nonumber \\
\left.\qquad\qquad\qquad\qquad\qquad -\tilde z_2(1)\big((\cos(2\pi n t)-1)(e^{2\pi n t}+1)+\sin(2\pi n t)(e^{2\pi n t}-1)\big)\right],\nonumber 
\end{gather*} $\tilde z^1=\mathcal P^{-1}z^1$, and $n\in\mathbb N\setminus\{0\}$.

\end{theorem}

\begin{proof} Fix $n\in\mathbb N$ and choose any $\theta_1,\theta_2,\theta_3$ such that $\sqrt{\sum_{k=1}^2|\theta|_k}=2\pi n$. We have uncountably many triples that are parameterized by the 3-$d$ sphere. This choice defines the orthogonal transformation $\mathcal P$ given by~\eqref{mathcalP}. Given $z^1$ we find the auxiliarry parameters $\tilde z^1$ by $\tilde z^1=\mathcal P^{-1}z^1$. 

Setting the values of $\dot x_0(0)$, $v_2$ and $\dot x_
2(0),\dot x_3(0)$ from~\eqref{dotx0} and~\eqref{eq:z1} into the general solutions~\eqref{x-geod2} and~\eqref{eq:12} and applying the orthogonal transformation $\mathcal P$ to $\tilde z(t)$ we finish the proof.
\end{proof} We observe that projections of the geodesics onto $(x_0,x_1)$-plane are straight lines $x_0(s)=x_1(s)=s$. The projections onto $(x_2,x_3)$-plane are circles from the lemma \ref{circle}.

3.1.3) For the case $\dot x_0(0)=-\dot x_1(0)\neq0$ in the same way as above we obtain the following 
\begin{theorem}
Given a point $A=(x^1,z^1)$, $|x^1|_L=0$, $x_0(1)=-x_1(1)\neq0$, $|z^1|\neq 0$, there are uncountably many spacelike geodesics connecting the origin $O=(0,0)$ with $A$. The equations of geodesics are given by expressions 
\begin{gather*}
x_0(t)=-x_1(t)=x_0(1)\cdot\frac{e^{-2\pi nt}-1}{e^{-2\pi n}-1}\\
\quad x_2^2(t)+x_3^2(t)=\frac{-2\tilde z_1(1)}{\pi n}\sin^2(\pi nt)
\end{gather*}
and $z(t)=\mathcal P\tilde z(t)$, where
\begin{gather*}
\tilde z_1(t)=\tilde z_1(1)(t-\frac{\sin(2\pi nt)}{2\pi n}),\\
\tilde z_2(t)=\frac{1}{2(e^{-2\pi n}+1)}\left[-\tilde z_3(1)\big((\cos(2\pi n t)-1)(e^{-2\pi n t}-1)-\sin(2\pi n t)(e^{-2\pi n t}+1)\big)\right.\\
\left.\qquad\qquad\qquad\qquad\qquad +\tilde z_2(1)\big((\cos(2\pi n t)+1)(e^{-2\pi n t}+1)+\sin(2\pi n t)(e^{-2\pi n t}-1)\big)\right],\\
\tilde z_3(t)=\frac{1}{2(e^{-2\pi n}+1)}\left[\tilde z_3(1)\big((\cos(2\pi n t)+1)(e^{-2\pi n t}+1)+\sin(2\pi n t)(e^{-2\pi n t}-1)\big)\right.\\
\left.\qquad\qquad\qquad\qquad\qquad -\tilde z_2(1)\big((1-\cos(2\pi n t))(e^{-2\pi n t}-1)+\sin(2\pi n t)(e^{-2\pi n t}+1)\big)\right].
\end{gather*} $\tilde z^1=\mathcal P^{-1}z^1$, and $n\in\mathbb N\setminus\{0\}$.

\end{theorem}

3.2) Consider the case $v_1>0$, i.~e., the initial velocity $\dot{x}^0$ is spacelike. It suggests that the ratio $-\frac{v_1}{v_2}$ is negative, which contradicts to~\eqref{eq:ratio}. Hence, in this case there are no spacelike geodesics with the property~ \eqref{eq:boundary}.

3.3) Now suppose $v_1<0$. It implies that the ratio $-\frac{v_1}{v_2}$ is positive and here we have 3 more subcases:

3.3.1) $-\frac{v_1}{v_2}>1$. In this case the initial velocity appears to be timelike. There are no points where this can happen because  the function $\nu(|\theta|)=\frac{\sin^2\frac{|\theta|}{2}}{\sinh^2\frac{|\theta|}{2}}$ is wave decaying with exponential speed on the right halfplane and bounded with 0 from below and with 1 from above. Thus, there are no timelike geodesics with the considered timelike initial velocity.

3.3.2) $-\frac{v_1}{v_2}=1$. Then $\sin^2\frac{|\theta|}{2}=\sinh^2\frac{|\theta|}{2}$, which is possible if and only if $|\theta|=0$. This case corresponds to Case 1), when the initial velocity is lightlike.

3.3.3)  $0<-\frac{v_1}{v_2}<1$. Let us investigate the quantity of Hamiltonian geodesics in this case. 

Using \eqref{eq:zsquare} and the fact that $\dot x^0=W^{-1}(1)x^1$ we calculate that 
\begin{gather}\label{z2z3}
    \tilde z_2^2(1)+\tilde z_3^2(1)=\frac{(-x_0^2(1)+x_1^2(1))x_0^2(1)}{2|\theta|^2\sinh^2\frac{|\theta|}{2}\sin^2\frac{|\theta|}{2}}\big(\cosh|\theta|-\cos|\theta|-\sin|\theta|\sinh|\theta|\big), \\
   \tilde z_1(1)=\frac{-x_0^2(1)+x_1^2(1)}{4|\theta|^2\sinh^2\frac{|\theta|}{2}\sin^2\frac{|\theta|}{2}}
\big(\sin^2\frac{|\theta|}{2}(\sinh|\theta|-|\theta|)-\sinh^2\frac{|\theta|}{2}(\sin|\theta|-|\theta|)\big)\nonumber 
\end{gather}
and notice that the ratio $\frac{\tilde z_2^2(1)+\tilde z_3^2(1)}{-2\tilde z_1(1) x_0^2(1)}$ depends only on $|\theta|$ and 
\begin{equation}
 \label{eq:norm}
  \frac{\tilde z_2^2(1)+\tilde z_3^2(1)}{-2\tilde z_1(1) x_0^2(1)}=\frac{\cosh|\theta|-\cos|\theta|-\sin|\theta|\sinh|\theta|}{\sin^2\frac{|\theta|}{2}\cdot(\sinh|\theta|-|\theta|)-\sinh^2\frac{|\theta|}{2}\cdot(\sin|\theta|-|\theta|)}.
\end{equation}

\begin{lemma}
    The function
\begin{equation}\label{eq:mu}
     \mu(|\theta|)=\frac{\cosh|\theta|-\cos|\theta|-\sin|\theta|\sinh|\theta|}{\sin^2\frac{|\theta|}{2}\cdot(\sinh|\theta|-|\theta|)-\sinh^2\frac{|\theta|}{2}\cdot(\sin|\theta|-|\theta|)}
\end{equation}
is nonnegative, has countably many points $|\theta|_k$ where $\mu(|\theta|_k)$ vanishes and on each interval $\big(|\theta|_k,\,|\theta|_{k+1}\big)$, $k=0,1,2,\ldots$ the function $\mu$ has a unique critical point $m_k$. On each of these intervals the function $\mu$ strictly increases from $0$ to $\mu(m_k)$, and then, strictly decreases from $\mu(m_k)$ to $0$. 
\end{lemma}
\begin{proof}
    Let us denote the numerator $\cosh|\theta|-\cos|\theta|-\sin|\theta|\sinh|\theta|$ by $f(|\theta|)$ and study this function.
   Its derivative $f'(|\theta|)=4\sin\frac{|\theta|}{2}\sinh\frac{|\theta|}{2}\cos\frac{|\theta|}{2}\cosh\frac{|\theta|}{2}(\tan\frac{|\theta|}{2}-\tanh\frac{|\theta|}{2})$ turns into $0$ if 
$$ |\theta|=2\pi n, \;\;n=0,1,2,\ldots \quad\mbox{or}$$
\begin{equation}\label{eq:tan} \tan\frac{|\theta|}{2}=\tanh\frac{|\theta|}{2}.
\end{equation}
The latter equation has infinitely many solutions $|\theta|_k$, $k=0,1,2,\ldots$, where $|\theta|_0=0$ and all the other roots  $|\theta|_k\in((2k-1)\pi,(2k+1)\pi)$ and $|\theta|_k\geqslant 2k\pi$, $k=1,2,\ldots$.
 Since 
\begin{gather*}
     f(|\theta|_k)=
2\cosh^2\frac{|\theta|_k}{2}-2\cos^2\frac{|\theta|_k}{2}-4\cos^2\frac{|\theta|_k}{2}\cosh^2\frac{|\theta|_k}{2}\tan\frac{|\theta|_k}{2}\tanh\frac{|\theta|_k}{2}\\
      = 2\cosh^2\frac{|\theta|_k}{2}-2\cos^2\frac{|\theta|_k}{2}-2\sin^2\frac{|\theta|_k}{2}\cosh^2\frac{|\theta|_k}{2}-2\sinh^2\frac{|\theta|_k}{2}\cos^2\frac{|\theta|_k}{2}=0
\end{gather*}
 and $f(2\pi n)=2\sinh^2 (\pi n)$, $n=1,2,\ldots$, we conclude that function $f$ has minimums at $|\theta|_k$, maximums at points $2\pi n$, and it is no-negative for $|\theta |>0$, see
Figure~\ref{fig:f1}.
\begin{figure}[ht]
\begin{center}
 \includegraphics[width=45mm,height=45mm]{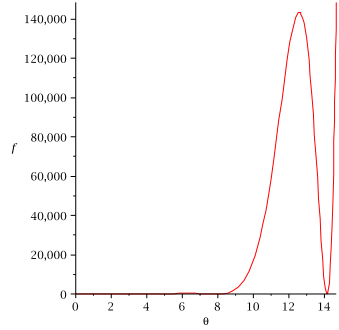}
\caption{The graph of function $f(\theta)$}\label{fig:f1}
\end{center}\end{figure}

The denominator $g(|\theta|)=\sin^2\frac{|\theta|}{2}(\sinh|\theta|-|\theta|)-\sinh^2\frac{|\theta|}{2}(\sin|\theta|-|\theta|)$  is an increasing function, passing through $0$ only at $|\theta|=0$, since the derivative $g'(|\theta|)=\frac{|\theta|}{2}(\sinh|\theta|-\sin|\theta|)+4\sinh^2\frac{|\theta|}{2}\sin^2\frac{|\theta|}{2}$ is always positive and equals to $0$ only at $|\theta|=0$. Hence, $g$ is positive for $|\theta|>0$. 

Since $g\neq0$ for $|\theta|\neq0$ and since $f(|\theta|_k)=0$, $k=1,2,\ldots$ the function $\mu=\frac{f}{g}$ equals to $0$ only at points $|\theta|_k$, $k=1,2,\ldots$ and is positive at all the other points $|\theta|\neq |\theta|_k$ and $|\theta|\neq0$. The point $|\theta|_0=0$ needs additional consideration.
Using Taylor decomposition of function $\mu$ near $|\theta|_0$ we get that $\mu(|\theta|)=\frac{1}{5}|\theta|+o(|\theta|^2)$, which goes to $0$ if $|\theta|\to0$. Therefore, $\mu(0)=0$.
Let us consider the derivative $\mu'(|\theta|)=\frac{f'(|\theta|)g(|\theta|)-f(|\theta|)g'(|\theta|)}{g^2(|\theta|)}$. We already know that $|\theta|_k$ are the solutions of $\mu'(|\theta|)=0$, since $f(|\theta|_k)=f'(|\theta|_k)=0$, therefore, they are local extremal points of function $\mu$. Since $\mu(|\theta|)$ is nonnegative for $|\theta|>0$ and $\mu(|\theta|_k)=\frac{f(|\theta|_k)}{|\theta|_k}=0$ we conclude that the points $|\theta|_k$ are minimums of function $\mu$.

Moreover, since $\mu$ is a smooth function on $(0,\infty)$, on each interval $\big(|\theta|_k,\,|\theta|_{k+1}\big)$, $k=0,1,2,\ldots$, it reaches its maximums, which we denote by $m_k$.

Let us make one more observation. Since $g$ is a monotonically increasing function, on each interval $[2\pi k, 2\pi(k+1)]$, $k=1,2,\ldots$ it has its local minimums at points $2\pi k$. Then, taking into account that $f$ has local maximums at points $2\pi n$, $n=1,2,\ldots$, we can estimate the function $\mu$ on each of these intervals from above by the value $\frac{f(2\pi(k+1))}{g(2\pi k)}=\frac{2\sinh^2(\pi (k+1))}{\pi k\cdot 2\sinh^2(\pi k)}$, $k=1,2,\ldots$. From here we obtain that $m_k\leqslant\frac{2\sinh^2(\pi (k+1))}{\pi k\cdot 2\sinh^2(\pi k)}\rightarrow 0$ with $k\to\infty$. Therefore, the function $\mu(|\theta|)$ decays to $0$ when $|\theta|$ goes to infinity. Note, that we can estimate $m_0$ from above by $\frac{f(2\pi)}{g(\pi)}=\frac{4\sinh^2\pi}{2\sinh\pi+\pi\cosh\pi-3\pi}$.

The graph of function $\mu$ is on the Figure~\ref{fig:f2}.
\begin{figure}[ht]\begin{center}
 \includegraphics[width=60mm,height=60mm]{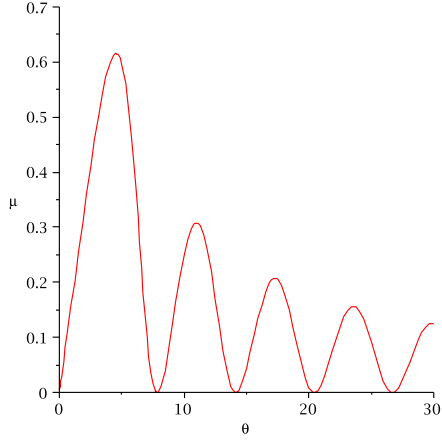}
\caption{The graph of function $\mu(\theta)$}\label{fig:f2}
\end{center}\end{figure}
\end{proof}

This lemma allows us to prove the following theorem. 
\begin{theorem}\label{general}
 Let $A=(x^1,z^1)$ be a point such that $|x^1|_L^2=0$, $x_1^2(1)<x_0^2(1)$ and $m_0$ is a global maximum of $\mu$. Suppose that $\left|\frac{z_2^2(1)+z_3^2(1)}{-2z_1(1) x_0^2(1)}\right|<m_0$ and that $|\theta|$ is a solution of the equation 
 \begin{equation}\label{mutheta}\left|\frac{z_2^2(1)+z_3^2(1)}{-2z_1(1) x_0^2(1)}\right|=\mu(|\theta|).\end{equation} Then there exist more than one spacelike geodesics joining the origin with~$A$.
 
  The equations of geodesics are given by 
\begin{gather}\label{xcoord}
   -x_0^2(t)+x_1^2(t)=\frac{-x_0^2(1)+x_1^2(1)}{\sinh^2\frac{|\theta|}{2}}\sinh^2\left(\frac{|\theta|t}{2}\right),\\
   x_2^2(t)+x_3^2(t)=\frac{x_2^2(1)+x_3^2(1)}{\sin^2\frac{|\theta|}{2}}\sin^2\left(\frac{|\theta|t}{2}\right) \nonumber
\end{gather}
and
\begin{equation}\label{zcoord}
\begin{split}
z_1(t)= \frac{-x_0^2(1)+x_1^2(1)}{4\sinh^2\frac{|\theta|}{2}}\big(\sinh(|\theta|t)-|\theta|t\big)+
\frac{x_2^2(1)+ x_3^2(1)}{4\sin^2\frac{|\theta|}{2}}\big(\sin(|\theta|t)-|\theta|t\big),\\
z_2^2(t)+z_3^2(t)  = \dfrac{x_2^2(1)+x_3^2(1)}{8\sin^2\frac{|\theta|}{2}\sinh^2\frac{|\theta|}{2}}\big(\cosh(|\theta|t)-\cos(|\theta|t) - \sinh(|\theta|t)\sin (|\theta|t) \big) \\
 \times \big((x_0^2(1)+x_1^2(1))\cosh(|\theta|t-|\theta|)+2x_0(1)x_1(1)\sinh(|\theta|t-|\theta|)-2x_0^2(1)+2x_1^2(1)\big),\end{split}\end{equation}

If $\left|\frac{z_2^2(1)+z_3^2(1)}{-2z_1(1) x_0^2(1)}\right|>m_0,$ then there are no geodesics of any causal type joining $0$ and $A$.
\end{theorem}

\begin{proof} Given the values $(x^1,z^1)$ of the end point we find $|\theta|$ as a solution to the equation~\eqref{mutheta}. For any of these solutions we get for $t=1$ from the expression $W^{-1}(t)x(t)=\dot{x}(0)$ 
\begin{gather}\label{velocity11}
   \dot x_0(0)=\frac{|\theta|}{2}\left(\frac{x_0(1)+x_1(1)}{e^{|\theta|}-1}+\frac{-x_0(1)+x_1(1)}{e^{-|\theta|}-1}\right),\nonumber \\
   \dot x_1(0)=\frac{|\theta|}{2}\left(\frac{x_0(1)+x_1(1)}{e^{|\theta|}-1}-\frac{-x_0(1)+x_1(1)}{e^{-|\theta|}-1}\right),\\
   \dot x_2(0)=\frac{|\theta|}{2}\left(\frac{x_2(1)-ix_3(1)}{e^{-i|\theta|}-1}+\frac{x_2(1)+ix_3(1)}{e^{i|\theta|}-1}\right),\nonumber \\
   \quad\dot x_3(0)=-\frac{i|\theta|}{2}\left(\frac{x_2(1)-ix_3(1)}{e^{-i|\theta|}-1}-\frac{x_2(1)+ix_3(1)}{e^{i|\theta|}-1}\right).\nonumber 
\end{gather}
From here we calculate
\begin{gather}\label{velocity}
v_1=    -\dot{x}_0^2+\dot{x}_1^2=-\frac{|\theta|^2 \big(-x_0^2(1)+x_1^2(1)\big)}{4\sinh^2\frac{|\theta|}{2}},\\
 v_2=    \dot{x}_2^2+\dot{x}_3^2=\frac{|\theta|^2 \big(x_2^2(1)+x_3^2(1)\big)}{4\sin^2\frac{|\theta|}{2}},\nonumber
\end{gather}
\begin{gather}\label{velocity1}
     \dot{x}_0^2+\dot{x}_1^2=\frac{|\theta|^2 }{16\sinh^2\frac{|\theta|}{2}}\left(e^{-|\theta|}(x_0(1)+x_1(1))^2+e^{|\theta|}(-x_0(1)+x_1(1))^2\right),\\
     \dot{x}_0\dot{x}_1=\frac{|\theta|^2 }{16\sinh^2\frac{|\theta|}{2}}\left(e^{-|\theta|}(x_0(1)+x_1(1))^2-e^{|\theta|}(-x_0(1)+x_1(1))^2\right).
\end{gather}

We change in general equations the values of $\frac{v_1}{|\theta|}$ and $\frac{v_2}{|\theta|}$ by~\eqref{velocity} and get~\eqref{xcoord}.

Next, for each solution to~\eqref{mutheta} we find an orthogonal transformation~\eqref{mathcalP} of $z$-space that fixes $z_1$-coordinate and therefore leave the expression $z_2^2(t)+z_3^2(t)=\tilde z_2^2(t)+\tilde z_3^2(t)$ invariant. Therefore using~\eqref{eq:zsquare}, we substitute there the necessary combinations by~\eqref{velocity} and~\eqref{velocity1} and get~\eqref{zcoord}.

We observe that according to Lemma~\ref{circle} projections of geodesics on $(x_2,x_3)$-plane are circles with center at $\left(-\frac{\dot{x}_3(0)}{|\theta|},\frac{\dot{x}_2^0}{|\theta|}\right)$ and radius $\frac{\sqrt{v_2}}{|\theta|}$. The projections on $(x_0,x_1)$-plane are suitable brunches of hyperbolas passing through $(0,0)$, see Lemma~\ref{hyperbola}. The parameters of the circles and hyperbolas can be rewritten in terms of the final point of a geodesic by making use of~\eqref{velocity11}.

Notice that if $z_2^2(1)+ z_3^2(1)\to 0$, the quantity of solutions to~\eqref{mutheta} is growing and the values of solutions tends to $2\pi k$. In this case~\eqref{z2z3} shows that the part of the velocity $v_1$ tends to $0$ and we come into particular situation of the Theorem~\ref{theorem2}.
\end{proof}

\begin{remark}
 Theorem~\ref{general} gives an answer about existence of spacelike geodesics and allows us estimate their cardinality. Unfortunately the estimation is not complete since we use only the orthogonal transformations in $z$-space leaving invariant $z_1$-coordinate. The difficulty is pure technical, since it is complicate to find the relation between the coordinates of finite point and the values of $|\theta|$. In the case of orthogonal transformations in $z$-space leaving invariant $z_1$-coordinate the relation is expressed by equation~\eqref{mutheta}.
\end{remark}

\begin{remark}
  Given a point $A=(x^1,z^1)$ on the surface $|x^1|_L^2=0$, $z^1\neq 0$, there are no timelike geodesics joining $0$ to $A$.
\end{remark}

\subsection{Connectivity between $(0,z^1)$ and $(x^1,z^1)$, where $z^1=const$}

Let us use the hyperspherical coordinates $(r,\varphi,\xi_1,\xi_2)$:
\begin{gather*}
    x_0+ix_1=re^{i\xi_1}\cos\frac{\varphi}{2},\quad x_2+ix_3=re^{i\xi_2}\sin\frac{\varphi}{2}
\end{gather*}
Then the horizontality conditions \eqref{eq:RH3} become
\begin{gather*}
 \dot{z}_1=-\frac{1}{2}r^2\big(\dot{\xi}_1\cos^2\frac{\varphi}{2}+\dot{\xi}_2\sin^2\frac{\varphi}{2}\big),\\
 \dot{z}_2=\frac{1}{4}r^2\big(\dot{\varphi}\sin(\xi_1+\xi_2)+\sin\varphi\cos(\xi_1+\xi_2)(\dot{x}_2-\dot{x}_1)\big),\\
 \dot{z}_3=\frac{1}{4}r^2\big(\dot{\varphi}\cos(\xi_1+\xi_2)+\sin\varphi\sin(\xi_1+\xi_2)(\dot{x}_1-\dot{x}_2)\big).
\end{gather*} Then,
\begin{equation*}
    \dot{z}_1^2+\dot{z}_2^2+\dot{z}_3^2=\frac{1}{4}r^2\big(\frac{\dot{\varphi}^2}{4}+\dot{\xi}_1^2\cos^2\frac{\varphi}{2}+\dot{\xi}_2^2\sin^2\frac{\varphi}{2}\big).
\end{equation*}
\begin{theorem}
     A smooth curve $c(t)$ is horizontal with constant $z$-coordinates if and only if $c(t)$ is a straight line in a 4-dimensional affine subspace: $c(t)=(\alpha_0t,\ldots,\alpha_3t,\,z_1^1,\,z_2^1,\,z_3^1)$ with $\alpha_0,\ldots\alpha_3\in\mathbb{R}^4$ and $\alpha_0^2+\ldots\alpha_3^2\neq0$.
\end{theorem}
\begin{proof}
Let $c(t)$ be a horizontal curve with constant vertical components. Then the equation $|\dot{z}|^2=\frac{1}{4}r^2\big(\frac{\dot{\varphi}^2}{4}+\dot{\xi}_1^2\cos^2\frac{\varphi}{2}+\dot{\xi}_2^2\sin^2\frac{\varphi}{2}\big)=0$ implies
$$\dot{\varphi}=0,\quad \dot{\xi}_1\cos\frac{\varphi}{2}=0,\quad \dot{\xi}_2\sin\frac{\varphi}{2}=0.$$
>From the first equation we conclude that $\varphi=\varphi^0$ is constant. In the case when $\varphi^0\neq\pi+2\pi k$ and $\varphi^0\neq 2\pi n$, $k,\,n\in\mathbb{Z}$, we see that $\xi_1$ and $\xi_2$ are constants $\xi_1^0$ and $\xi_2^0$. Therefore, horizontal components of the curve $c(t)$ are of the form 
\begin{gather*}
     x_0=t\cos\xi_1^0\cos\frac{\varphi^0}{2},\;x_1=t\cos\xi_2^0\sin\frac{\varphi^0}{2},\;x_2=t\sin\xi_1^0\cos\frac{\varphi^0}{2},\;x_3=t\sin\xi_2^0\sin\frac{\varphi^0}{2}.
\end{gather*}
If $\varphi^0=\pi+2\pi k$, $k\in \mathbb{Z}$, then
\begin{gather*}
     x_0=0,\;x_1=\pm t\cos\xi_2^0,\;x_2=0,\;x_3=\pm t\sin\xi_2^0,
\end{gather*} and if $\varphi^0=2\pi n$, $n\in\mathbb{Z}$, then
\begin{gather*}
     x_0=\pm t\cos\xi_1^0,\;x_1=0,\;x_2=\pm t\sin\xi_1^0,\;x_3=0.
\end{gather*}
Conversely, let us assume that $c(t)=(\alpha_0t,\ldots,\alpha_3t,z_1^1,z_2^1,z_3^1)$, where $z_1^1,\ldots,z_3^1$ are some constants. Set $\alpha t=(\alpha_0t,\ldots,\alpha_3t)$. Observe, that $(I_1\alpha,\alpha)=(I_2\alpha,\alpha)=(I_3\alpha,\alpha)=0$ for any vector $\alpha=(\alpha_0,\ldots,\alpha_3)$. Then,
\begin{gather*}
    \dot{\tilde z}_1=0=\frac{1}{2}(I_1(\alpha t),\dot{(\alpha t)})=\frac{t}{2}(I_1\alpha,\alpha),\quad
    \dot{\tilde z}_2=0=\frac{1}{2}(I_2(\alpha t),\dot{(\alpha t)})=\frac{t}{2}(I_2\alpha,\alpha),\\
    \dot{\tilde z}_3=0=\frac{1}{2}(I_3(\alpha t),\dot{(\alpha t)})=\frac{t}{2}(I_3\alpha,\alpha).
\end{gather*} This implies that $z$-coordinates are constants.
\end{proof}

\section{Reachable sets by geodeics on Lorentzian Heisenberg group}\label{sec:6}

In this section we would like to compare the results obtained for $\mathbb{H}$-type Quaternion group with Lorentzian metric and 3-dimensional Lorentzian Heisenberg group $\mathbb{H}^1_L$.

We remind that $\mathbb{H}^1_L$ is a triple $(\mathbb{R}^3,H,g)$, where $\mathbb{R}^3$ is equipped with noncommutative multiplication law
$$(x_0,x_1,z)\circ(x_0',x_1',z')=\big(x_0+x_0',x_1+x_1',z+z'+\frac12(x_1x_0'-x_0x_1')\big),$$
the subbundle $H$ is a span of two left invariant vector fields $X_0=\frac{\partial}{\partial x_0}+\frac{1}{2}x_1\frac{\partial}{\partial z}$, $X_1=\frac{\partial}{\partial x_0}-\frac{1}{2}x_0\frac{\partial}{\partial z}$, for which $[X_0,X_1]=Z=\frac{\partial}{\partial z}$, and $g$ is a Lorentzian metric on $H$ defined by
$$g(X_0,X_0)=-1,\quad g(X_1,X_1)=1,\quad g(X_0,X_1)=0.$$

In \cite{KM} authors investigated the connectivity by geodesics in $\mathbb{H}^1_L$. In particular, they obtained the following result.
\begin{theorem}
 Let $A=(x,y,z)$ be a point such that $x>0$ $(x<0)$, $-x^2+y^2<0$, $\frac{4|z|}{x^2-y^2}<1$. Then there is a unique future directed (past directed) geodesic, joining $O=(0,0,0)$ with the point $A$. Let $\theta$ be a solution of the equation
\begin{equation*}
 \frac{4z}{-x^2+y^2}=\frac{|\theta|/2}{\sinh^2(|\theta|/2)}-\coth(|\theta|/2).
\end{equation*}
Then the equations of timelike future directed geodesic $\gamma:[0,1]\to \mathbb{H}^1_L$ are
\begin{equation*}
 \begin{split}
    x(t)=\sinh^2(\frac{|\theta|}{2}t)(x(\coth(\frac{|\theta|}{2}t)\coth(\frac{|\theta|}{2})-1)+y(\coth(\frac{|\theta|}{2}t)-\coth(\frac{|\theta|}{2}))),\\y(t)=\sinh^2(\frac{|\theta|}{2}t)(y(\coth(\frac{|\theta|}{2}t)\coth(\frac{|\theta|}{2})-1)+x(\coth(\frac{|\theta|}{2}t)-\coth(\frac{|\theta|}{2}))),\\z(t)=z\frac{|\theta|t-\sinh(|\theta|t)}{|\theta|-\sinh(|\theta|)}.\hspace{83mm}
 \end{split}
\end{equation*}
\end{theorem}

Moreover, the authors obtained the following result about the reachability by causal Hamiltonian geodesics.
\begin{theorem}
 Let us define the following sets
\begin{gather*}
    R_t=\{-x^2+y^2<0,\;\frac{4|z|}{x^2-y^2}<1\},\\
    R_{sp}=\{-x^2+y^2>0,\;\frac{4|z|}{-x^2+y^21}<1\},\\
    R_l=\{-x^2+y^2=0,\;z=0\}.
\end{gather*}
Then there exists a unique geodesic connecting the point $O=(0,0,0)$ with a point $P$ that belongs to one of the sets $R_t$, $R_{sp}$ or $R_l$. Particularly, if $P\in R_t$, then the geodesic is timelike, if $P\in R_{sp}$, then the geodesic is spaselike, and if $P\in R_l$, then the geodesic is lightlike.
\end{theorem}
The connectivity on $\mathbb{H}^1_L$ is dependent on the solutions $|\theta|$ of the equation $$\frac{4z}{-x^2+y^2}=\mu(\theta),$$ where the function $\mu(\tau)=\frac{\tau/2}{\sinh^2(\tau/2)}-\coth(\tau/2)$ is strictly decreasing on the interval $(-\infty,\infty)$ from $-1$ to $1$. It means that if the point $A=(x,y,z)$  is such that $\left|\frac{4z}{-x^2+y^2}\right|\geqslant 1$, i.~e. if $A\in (R_t\cup R_{sp}\cup R_l)^c$, then there are no geodesics of any causal type joining the origin with $A$. In particular, there are no Hamiltonian geodesics joining the origin with the points of the surfaces $\{\frac{4|z|}{x^2-y^2}=1, -x^2+y^2<1\}$ and $\{\frac{4|z|}{-x^2+y^2}=1,-x^2+y^2>0\}$. We observe that in \cite{Gr5} it is shown that it is possible to find nonHamiltonian lightlike geodesics joining $0$ with the points of these surfaces.

Thus, in contrast with $\mathbb{Q}_L$, the Heisenberg Lorentzian group has the property of uniqueness of geodesics starting from the origin with given tangent vector. It happens due to lower dimension of the \textquotedblleft spacelike\textquotedblright part of $\mathbb{H}^1_L$.

\end{document}